\documentclass[journal,12pt,onecolumn,draftclsnofoot,]{IEEEtran}

\usepackage[english]{babel}
\usepackage[style=ieee,backend=bibtex]{biblatex}
\AtBeginBibliography{\small}

\usepackage[table]{xcolor}
\usepackage[utf8]{inputenc} 
\usepackage[T1]{fontenc}    
\usepackage{hyperref}       
\usepackage{url}            
\usepackage{booktabs}       
\usepackage{amsfonts,amsthm}       
\usepackage{nicefrac}       
\usepackage{microtype}      
\usepackage{color}
\usepackage{amsmath,array}
\usepackage{amssymb}
\usepackage{bm}
\usepackage{graphicx}
\usepackage[ruled,vlined,noresetcount,linesnumbered]{algorithm2e}
\usepackage{subcaption}
\usepackage[font=small,labelfont=bf,margin=1cm]{caption}

\usepackage{xcolor}

\usepackage[normalem]{ulem}

\theoremstyle{plain}
\newtheorem{thm}{Theorem}[section]
\newtheorem{lem}[thm]{Lemma}
\newtheorem{prop}[thm]{Proposition}
\newtheorem{corollary}[thm]{Corollary}

\theoremstyle{definition}
\newtheorem{definition}{Definition}[section]

\newtheorem{example}{Example}[section]
\newtheorem*{remark}{Remark}

\newcommand{\supp}{\mathrm{supp}}

\newcommand{\Mod}[1]{\ (\mathrm{mod}\ #1)}

\def\*#1{\mathbf{#1}}

\title{Modular Polynomial Codes for Secure and Robust Distributed Matrix Multiplication}
\author{David Karpuk and Razane Tajeddine
\thanks{D.\ Karpuk is with W/Intelligence, WithSecure Corporation, Helsinki, Finland and the Department of Mathematics and Systems Analysis, Aalto University, Espoo, Finland.  email: \texttt{davekarpuk@gmail.com}.}
\thanks{R.\ Tajeddine is with the Department of Computer Science, University of Helsinki, Helsinki, Finland.  email: \texttt{razane.tajeddine@helsinki.fi}.}
\thanks{The work of R. Tajeddine was supported by the Research Council of Finland under grant 343555.}}

\date{\today}

\begin{document}

\maketitle

\begin{abstract}
We present Modular Polynomial (MP) Codes for Secure Distributed Matrix Multiplication (SDMM).  The construction is based on the observation that one can decode certain proper subsets of the coefficients of a polynomial with fewer evaluations than is necessary to interpolate the entire polynomial.  We also present Generalized Gap Additive Secure Polynomial (GGASP) codes.  Both MP and GGASP codes are shown experimentally to perform favorably in terms of recovery threshold when compared to other polynomials codes for SDMM which use the grid partition.  Both MP and GGASP codes achieve the recovery threshold of Entangled Polynomial Codes for robustness against stragglers, but MP codes can decode below this recovery threshold depending on the set of worker nodes which fails.  The decoding complexity of MP codes is shown to be lower than other approaches in the literature, due to the user not being tasked with interpolating an entire polynomial.
\end{abstract}

\section{Introduction}

\subsection{Background and Motivation}

Multiplication of large matrices is a fundamental building block of modern science and technology.  Matrix multiplication is especially ubiquitous now in the era of Big Data, as the training and inference of high-powered Machine Learning models requires performing linear operations on data sets with millions or even billions of records.  Such computations are often so intensive that they cannot be performed by a single machine, but are rather outsourced to a computing cluster in the cloud.  On the other hand, Machine Learning models are often trained on personal user data, and respecting the privacy needs of the individuals to whom the training data belongs may require revealing as little information as possible about the data to be processed to those doing the processing.

The problem of \emph{Secure Distributed Matrix Multiplication} (SDMM) considers a user who wishes to compute the product of two (large) block matrices $A$ and $B$, partitioned as
\begin{equation}\label{basic_matrices}
A = \begin{bmatrix}
    A_{0,0} & \cdots & A_{0,M-1} \\
    \vdots & \ddots & \vdots \\
    A_{K-1,0} & \cdots & A_{K-1,M-1}
\end{bmatrix}
\quad\text{and}\quad
B = \begin{bmatrix}
    B_{0,0} & \cdots & B_{0,L-1} \\
    \vdots & \ddots & \vdots \\
    B_{M-1,0} & \cdots & B_{M-1,L-1}
\end{bmatrix}.
\end{equation}
This massive matrix multiplication is distributed to $N$ worker nodes, each of which performs a single submatrix multiplication of the same size as one of the sub-products $A_{k,m}B_{m',\ell}$.  The parameters $K$, $M$, and $L$ therefore control the complexity of the operation performed at each worker node.  Given a realistic scenario with $N$ worker nodes, each with a fixed computational budget, the user may have good reason to fix all three parameters to be some desired values.  Concerned about leaking sensitive information to the worker nodes, the user wishes to perform this task under the constraint of $T$-security, wherein no subset of $T\leq N$ worker nodes may learn any information about the matrices $A$ and $B$.

We will refer to the case where $M = 1$ as the \emph{outer product partition} and the case where $K = L = 1$ as the \emph{inner product partition}.  These special cases have been dealt with extensively in the literature, as reviewed below.  The case where all three parameters are greater than one will be called the \emph{grid partition}.  Grid partitioning appears necessary in distributed matrix multiplication if one has limited computational budgets at each worker node, and wants to develop protocols for input matrices of any size whatsoever.

A related question is that of \emph{robustness} in distributed computation.  In a large, distributed computing system, it is natural that one or more worker nodes may be delayed in returning the result of their computation to the user, or may not return any computation at all due to, e.g., being disconnected from the network.  Such workers are called \emph{stragglers}, and coding-theoretic techniques are often introduced to deal with this phenomenon as one would deal with an erasure using an error-correcting code.  

The problem under consideration is thus briefly stated as follows.  Given a partitioning of two matrices $A$ and $B$ as in \eqref{basic_matrices}, and some desired levels of security and robustness, what is the minimum number of worker nodes $N$ needed to guarantee a successful computation of the product $AB$?  This number $N$ is often referred to as the \emph{recovery threshold}. 
 To normalize $N$ with respect to $K$, $M$, and $L$, performance is often measured in terms of the \emph{rate} $KML/N$.

\subsection{Related Work}

Ideas from Coding Theory have been used for several years now to provide robustness against stragglers in distributed computation frameworks \cite{high_dim_coded,high_dim_coded_journal,limited_sharings,secure_coded_mp}.  Polynomial codes were introduced in \cite{polycodes} as a method of mitigating stragglers in the distributed multiplication of large matrices.  MatDot and PolyDot codes \cite{straggler_optimal} and Entangled Polynomial Codes \cite{entangled1, entangled2, straggler_fund} were introduced shortly thereafter, which further improved upon the constructions in \cite{polycodes}, and offered trade-offs between upload cost, download cost, and recovery threshold.  These trade-offs were further studied in \cite{tradeoff}, MatDot codes of \cite{straggler_optimal} were applied to $k$-nearest neighbors estimation in \cite{k_means}, and PolyDot codes were applied to neural network training in \cite{deep_nn_coded}.  Recent progress in this area was done in \cite{duursma} by constructing distributed versions of Strassen's algorithm and other lower-rank matrix multiplication algorithms.

The problem of SDMM for the outer product partition originated with \cite{ravi}, and was followed by \cite{kakar}, wherein the authors gave more rate-efficient schemes.  The Gap Additive Secure Polynomial (GASP) codes \cite{d2020gasp, d2021degree} again showed drastic improvements in rates over previous constructions, by constructing polynomial codes with gaps in the degrees of the polynomials to be interpolated; the user then knows \emph{a priori} that some of the coefficients to be recovered are zero.

Since the above work there has been much progress on SDMM, especially in connection with matrix computations on coded data \cite{jie_cam_coded, koreans_coded, mds_coded_pmm} and studying trade-offs between upload cost, download cost, encoding and decoding complexity, latency, privacy, and security \cite{bivariate, rawad_adaptive, rawad_latency, upload_vs_download, systematic_private}.  The cross-subspace alignment codes of \cite{jafar_batch, jafar_batch2} have deep connections with Private Information Retrieval (PIR) \cite{jafar_cross, jafar_xsecure_tprivate} and Lagrange Coded Computing \cite{lagrange}.  Transforming an SDMM problem into a PIR problem as done in \cite{jafar_batch} essentially requires a shift in the objective function; in the current work we wish to minimize the number of worker nodes $N$ given fixed partition variables $K$, $M$, and $L$ (and therefore, fixed computational power at each worker node).  However, in \cite{jafar_batch}, in addition to a different system model, the number of worker nodes $N$ is assumed to be fixed and questions of capacity are studied under this constraint.  The connection and correlation between the recovery threshold and various measures of communication cost has been elaborated on in \cite{rafael_note, inner_product}.  The Discrete Fourier Transform codes of \cite{inner_product} achieve the best-known recovery threshold for the inner product partition, namely $M + 2T$.

Our main points of comparison are \cite{root_of_unity, oliver, flex}, which focus on minimizing the number of worker nodes necessary under the condition of $T$-security using the grid partition.  The results of \cite{root_of_unity} generalize the construction of \cite{inner_product} to the grid partition.  Underpinning the current work as well as \cite{inner_product, root_of_unity} is the fact that the sum of all $M^{th}$ roots of unity is zero, but that appears to be where the similarities end.  The authors of \cite{oliver} show experimentally that their codes outperform those of \cite{flex} in terms of recovery threshold for a number of parameters; therefore to keep the presentation of our results tractable we omit direct comparison with \cite{flex} and rather focus on \cite{root_of_unity, oliver} in our experiments.

One promising avenue of future research is the application of Algebraic Geometry codes to SDMM.  The authors of \cite{okko4} show that Algebraic Geometry codes constructed from Kummer extensions of function fields allow for rates comparable to those of GASP codes.  In \cite{hermitian}, the authors use codes constructed from Hermitian curves to reduce the necessary field size of certain SDMM protocols.  Algebraic tools have proved their worth in other ways as well, as the authors of \cite{field_trace} design codes based on the field trace to reduce total communication load using the literature on decoding Reed-Solomon codes.

\subsection{Outline and Summary of Main Contributions}

The main contributions of the current work are as follows:
\begin{enumerate}
    \item In Section \ref{partial_poly} we introduce the mechanism of \emph{partial polynomial interpolation}, which allows one to decode certain proper subsets of the non-zero coefficients of a polynomial whose indices satisfy a specified congruence condition, with fewer evaluations than is necessary to interpolate the entire polynomial.  The central tool is what we refer to as the $\mathrm{mod}$-$M$ \emph{transform} of a polynomial; see Definition \ref{transform}.
    \item In Section \ref{mpc} we use the tools of the previous section to define \emph{Modular Polynomial} (MP) codes for SDMM; see Definition \ref{mpc_defn} and the subsequent discussion.  Theorems \ref{mpc_decodability_security} and \ref{polynomials} and Corollary \ref{finally} study sufficient conditions for the decodability and security of such codes.  Subsection \ref{mpc_exmaple} concludes the section with an explicit example. Section \ref{mpc_rate_section} is devoted to Theorem \ref{mpc_rate}, which calculates the recovery threshold of MP codes as an explicit function of the system parameters, for a certain subclass of such codes.  
    \item Section \ref{ggasp_section} defines GGASP codes, a generalization of the GASP codes of \cite{d2020gasp,d2021degree} to the grid partition.  An explicit formula for the recovery threshold of such codes is given in Theorem \ref{ggasp_r_rate} as a function of the system parameters, and the section is concluded with an explicit example.
    \item Section \ref{comparison} provides several experiments which compare MP and GGASP codes with the existing SDMM codes in the literature for the grid partition.  We show these codes compare favorably in a number of system parameter regimes.
    \item In Section \ref{robustness_section} we generalize our system model to include straggler worker nodes.  We show in Theorem \ref{mp_recovery_threshold} and Theorem \ref{ggasp_robust_t} that MP and GGASP codes achieve the state-of-the-art recovery threshold in the absence of $T$-security.  Additionally, MP codes are shown to be able to decode below the recovery threshold depending on which subsets of worker nodes fail to return responses to the user, a property we precisely quantify with Theorem \ref{prob_recovery}.  Theorems \ref{ggasp_robust_t} and \ref{mp_recovery_t} describe the ability of MP and GGASP codes to provide $T$-security and robustness simultaneously, and we conclude the section with a detailed example.
    \item In Section \ref{complexity} we briefly study the computational complexity of MP and GGASP codes.  The decoding complexity of MP codes compares especially favorably to previous constructions in the literature, as the user only has to decode $N/M$ coefficients of a certain polynomial, instead of $N$ as is almost always the case in the SDMM literature.
    \item We conclude in Section \ref{conclusion} and provide a short discussion of future research directions.
\end{enumerate}

\subsection{Notation}

We use the symbols $\mathbb{Z}$ and $\mathbb{Z}_{\geq 0}$ to denote the integers and non-negative integers, respectively.  Matrices will generally be denoted by capital letters such as $A$ and $B$.  We use $\mathbb{F}$ and $\mathbb{K}$ to denote fields, $\mathbb{F}^\times$ to denote the multiplicative group of non-zero elements of $\mathbb{F}$, and for a prime power $q$ we use the notation $\mathbb{F}_q$ to denote the finite field of cardinality $q$.  For any integers $x$ and $y$, we define
\[
[x:y] := \left\{
\begin{array}{cl}
\{x,x+1,\ldots,y\} & \text{if } x\leq y \\
\emptyset & \text{if } x>y
\end{array}
\right.
\]
We adopt the following conventions for any sequence $\{x_i\}$ of real numbers and any integer-indexed collection $\{\mathcal{A}_i\}$ of sets: $\sum_{i = a}^b x_i = 0$ if $a>b$, and $\bigcup_{i = a}^b\mathcal{A}_i = \emptyset$ if $a>b$.  Such conventions serve only to streamline some formulas.  If $\mathcal{A}$ and $\mathcal{B}$ are any two sets, we denote their disjoint union with $\mathcal{A}\sqcup\mathcal{B}$.

\section{Partial Polynomial Interpolation and the mod-$M$ Transform} \label{partial_poly}

In contrast to all previous work on SDMM, our decoding algorithm for polynomial codes allows a user to decode a proper subset of coefficients of the desired polynomial, with fewer evaluations than there are non-zero coefficients of the polynomial.  We refer to this general process as \emph{partial polynomial interpolation}, and this section is devoted to the mathematics underlying this tool.  

\subsection{Generalized Vandermonde Matrices}

We wish to generalize the process of polynomial interpolation to recover certain subsets of coefficients of polynomials, which necessitates some basic understanding of generalized Vandermonde matrices.  Recall that a $P\times N$ matrix has the \emph{MDS property} if every $P\times P$ minor is invertible.  Equivalently, it is the generator matrix of an MDS code.  {We will make extensive use of the following Lemma.}

\begin{lem}\label{sz_non_vanishing}
    {Let $\mathbb{F}$ be a field and let $\phi_1,\ldots,\phi_S$ be non-zero polynomials in $\mathbb{F}[X_1,\ldots,X_N]$.  Then there exists a finite extension $\mathbb{K}/\mathbb{F}$ and a point $a\in \mathbb{K}^N$ such that $\phi_s(a)\neq 0$ for all $s$.}
\end{lem}
\begin{proof}
    {This result follows from basic notions in Algebraic Geometry, but we give a more elementary proof using the Schwartz-Zippel Lemma as in the proof of \cite[Theorem 1]{d2020gasp}.}
    
    {Let $D$ be the maximum of the total degrees of the $\phi_s$.  Choose an extension $\mathbb{K}/\mathbb{F}$ such that $|\mathbb{K}| > SD$ and let $U\subseteq \mathbb{K}$ be a set of size $SD$.  Sample the entries $a_n$ of $a$ independently and uniformly from $U$, and let $\mathsf{Z}$ be the union of the events $\phi_s(a) = 0$.  To prove the Lemma it suffices to show that $\mathrm{Pr}(\mathsf{Z})< 1$. By the union bound and the Schwartz-Zippel Lemma, we have}
    \[
    {\mathrm{Pr}(\mathsf{Z}) \leq \sum_{s = 1}^S \mathrm{Pr}(\phi_s(a)=0) \leq S\frac{D}{|U|} < 1}
    \]
    {which completes the proof of the Lemma.}
\end{proof}

\begin{definition}
    Let $P\leq N$ be positive integers, and let $\mathcal{I} = \{i_1,\ldots,i_P\}$ be a set of non-negative integers such that $0\leq i_1< \cdots < i_P$.  Let $\mathbb{F}$ be a field and let $a = (a_1,\ldots,a_N)\in \mathbb{F}^N$.  We define the \emph{generalized Vandermonde matrix} associated to this data to be the $P\times N$ matrix
    \[
    GV(a,\mathcal{I}) :=  \begin{bmatrix}
    a_n^{i_p}
    \end{bmatrix}_{\substack{1\leq p\leq P \\ 1\leq n\leq N}} \in \mathbb{F}^{P\times N}.
    \]
\end{definition}

\begin{prop}\label{evaluation}
    Let $\mathbb{F}$ be a field.  Then for any positive integers $P\leq N$ and any set of non-negative integers $\mathcal{I} = \{i_1,\ldots,i_P\}$ such that $0\leq i_1< \cdots < i_P$, there exists a finite extension $\mathbb{K}/\mathbb{F}$ and a vector $a\in \mathbb{K}^N$ such that $GV(a,\mathcal{I})$ has the MDS property.
\end{prop}
\begin{proof}
    {It is a} well-known fact that the determinants of the $P\times P$ submatrices of $GV(a,\mathcal{I})$ are non-zero polynomials in the ring $\mathbb{F}[a_1,\ldots,a_N]$, when one considers the $a_n$ as variables. {Now we apply Lemma \ref{sz_non_vanishing} to conclude the proof.}
\end{proof}

{In what follows, we explore the necessity of passing to a non-trivial finite extension $\mathbb{K}/\mathbb{F}$ of the base field so that the conclusion of the above proposition holds.  Specifically, see Proposition \ref{field_size}, Example \ref{ppi_example}, and Subsection \ref{mpc_exmaple}.  For both Lemma \ref{sz_non_vanishing} and Proposition \ref{evaluation} we could have sidestepped the notion of field extension and simply assumed that the base field $\mathbb{F}$ was sufficiently large.  However, the current presentation emphasizes that the polynomials themselves do not need to be defined over a sufficiently large field, only the evaluation vector does.}

\subsection{Partial Polynomial Interpolation}

Let $V$ be a finite-dimensional vector space over a field $\mathbb{F}$; in our case $V$ will be the set of matrices over $\mathbb{F}$ of a fixed size.  We define the set of polynomials with coefficients in $V$ as follows:
\[
V[x] := \{h = h(x) = v_0 + v_1x + \cdots + v_nx^n\ |\ v_i \in V,\ n \geq 0,\ v_n\neq 0\}
\]
with the obvious notion of addition and scalar multiplication by elements of $\mathbb{F}$.  Any $a\in \mathbb{F}$ gives us an evaluation map $V[x]\rightarrow V$, defined by $h\mapsto h(a) = \sum_{i = 0}^n v_ia^i$.  If $h$ has degree $n$, and $a_0,\ldots,a_n$ are any $n+1$ pairwise distinct elements of $\mathbb{F}$, then we can recover the coefficients of $h$ from the evaluations $h(a_0),\ldots,h(a_n)$ using  polynomial interpolation.

With the end goal of generalizing polynomial interpolation to recover proper subsets of coefficients of polynomials $h\in V[x]$, we make the following definition.
\begin{definition}
    Let $h = v_0 + v_1x + \cdots + v_nx^n \in V[x]$ and let $M>0$ be a positive integer.  We define
\[
\supp(h) = \{i \ |\ v_i\neq 0\}\quad\text{and}\quad
    \supp_{-1}(h) = \{i\in \supp(h)\ |\ i\equiv M-1 \Mod{M}\}.
    \]
    We refer to $\supp(h)$ as the \emph{support} of $h$.
\end{definition}

\begin{definition}\label{transform}
Let $M>0$ be a positive integer and assume that $\mathbb{F}$ contains a primitive $M^{th}$ root of unity $\zeta$.  For $h \in V[x]$ define
\begin{equation}
\widehat{h} := \frac{1}{M}\sum_{m = 0}^{M-1}\zeta^m h(\zeta^m x)
\end{equation}
 to be the \emph{$\mathrm{mod}$-$M$ transform} of $h$.
\end{definition}

Note that the mod-$M$ transform is well-defined, since if $\mathbb{F}$ contains a primitive $M^{th}$ root of unity it must be the case that $M\neq 0$ in $\mathbb{F}$.  It is also worth remarking that $h\mapsto \widehat{h}$ is additive as a map from $V[x]$ to itself, so that $\widehat{h_1+h_2} = \widehat{h}_1 + \widehat{h}_2$ for any $h_1,h_2\in V[x]$.
\begin{prop}
Fix $h\in V[x]$ and $M$ as above.  Then we can express the mod-$M$ transform of $h$ as
\[
\widehat{h} = \sum_{i\in \supp_{-1}(h)} v_i x^i.
\]
In particular, $\supp(\widehat{h}) = \supp_{-1}(h)$.
\end{prop}
\begin{proof}
We can expand out $\widehat{h}$ to get
\begin{align*}
\widehat{h} &= \frac{1}{M}\sum_{m = 0}^{M-1}\zeta^{m}h(\zeta^mx) = \frac{1}{M}\sum_{m = 0}^{M-1}\zeta^{m}\left(
\sum_{i = 0}^nv_i\zeta^{mi}x^i
\right) \\
&= \frac{1}{M}\sum_{m = 0}^{M-1} \sum_{i = 0}^n v_i\zeta^{(i+1)m}x^i = \frac{1}{M}\sum_{i = 0}^n v_i \left(
\sum_{m = 0}^{M-1} \zeta^{(i+1)m}
\right) x^i.
\end{align*}
Now one can see that
\[
\sum_{m = 0}^{M-1}\zeta^{(i+1)m} = \left\{
\begin{array}{cl}
M & \text{if $i\equiv M-1 \Mod{M}$} \\
0 & \text{otherwise}
\end{array}
\right.
\]
as in the latter case this is a sum over all $L^{th}$ roots of unity for some $L|M$, and the sum of all $L^{th}$ roots of unity is zero.  Plugging this expression into $\widehat{h}$ completes the proof.
\end{proof}

\begin{prop}\label{recovery}
    Let $h = v_0 + v_1x + \cdots + v_nx^n\in V[x]$, let $M>0$ be a positive integer and suppose that $\mathbb{F}$ contains a primitive $M^{th}$ root of unity.  Then we can recover all of the $v_i$ such that $i\in\supp_{-1}(h)$ with $M |\supp(\widehat{h})|$ evaluations of $h(x)$ over some finite extension $\mathbb{K}/\mathbb{F}$.
\end{prop}
\begin{proof}
We will recover the desired coefficients of $h$ by interpolating the polynomial $\widehat{h}$.  For any evaluation point $a$ we can compute $\widehat{h}(a)$ by computing the $M$ evaluations $\zeta^mh(\zeta^ma)$ for $m = 0,\ldots,M-1$, summing the result over $m$, and dividing by $M$ to obtain $\widehat{h}(a)$.

Let $\supp(\widehat{h}) = \{i_1,\ldots,i_P\}$, and let $a = (a_1,\ldots,a_P)\in \mathbb{K}^P$, where $\mathbb{K}$ is some finite extension of $\mathbb{F}$.  The matrix $GV(a,\supp(\widehat{h}))$ is the $P\times P$ coefficient matrix of the linear system used to solve for the coefficients of $\widehat{h}$, using the evaluation points $a_1,\ldots,a_P$.  By {Proposition} \ref{evaluation} there exists an extension $\mathbb{K}/\mathbb{F}$ and {an evaluation vector} $a\in \mathbb{K}$ such that $\det(GV(a,\supp(\widehat{h})))\neq 0$.  The corresponding linear system can therefore be solved for the desired coefficients $v_i$ where $i\in \supp(\widehat{h})$.
\end{proof}

\begin{prop}\label{field_size}
    With the notation of the statement and proof of Proposition \ref{recovery}, suppose that $M>1$ and $P = |\supp(\widehat{h})|>1$.   Let $a\in \mathbb{K}^P$ and assume that $det(GV(a,\supp(\widehat{h})))\neq 0$.  Then:
    \begin{enumerate}
        \item $a_p\neq 0$ for all $p$, and $a_p^M\neq a_q^M$ for all distinct $p,q$
        \item $|\mathbb{K}|\geq MP + 1$

    \end{enumerate}
\end{prop}
\begin{proof}

    First we observe that we can write $\widehat{h}(x) = x^{M-1}k(x^M)$ where  $|\supp(k)| = |\supp(\widehat{h})| > 1$.  The respective generalized Vandermonde determinants are related by 
    \[
    \det(GV(a,\supp(\widehat{h}))) = \prod_{p=1}^P a_p^{M-1} \cdot \det(GV((a_1^M,\ldots,a_P^M),\supp(k))).
    \]
    This immediately implies part 1) of the proposition.
    
    Now assume that $|\mathbb{K}| \leq MP$.  The map $x\mapsto x^M$ is a homomorphism from $\mathbb{K}^\times$ to itself with kernel of size $M$, and therefore image of size $(|\mathbb{K}|-1)/M < P$.  Therefore within any subset of $\mathbb{K}^\times$ of size $P$ there exist distinct $a_p, a_q$ such that $a_p^M = a_q^M$.  This contradicts the above conclusion, which completes the proof.
\end{proof}

\begin{example}\label{ppi_example}
    Let $h = v_0 + v_1x + \cdots + v_6x^6\in V[x]$ where $V$ is a finite-dimensional vector space over $\mathbb{F}_{7}$.  Let us set $M = 3$, so that $\zeta = 2$ is a primitive $3^{rd}$ root of unity.  Then $\widehat{h} = v_2x^2 + v_5x^5$, so $\supp(\widehat{h}) = \{2,5\}$ and we can recover $v_2$ and $v_5$ with $M|\supp(\widehat{h})| = 6$ evaluations, whereas interpolating the entire polynomial $h$ would take $6+1 = 7$ evaluations.  Any evaluation vector $a = (a_1,a_2)$ over a finite extension of $\mathbb{F}_7$ will suffice, provided that the matrix
    \[
    GV(a,\supp(\widehat{h})) = \begin{bmatrix}
        a_1^2 & a_2^2 \\
        a_1^5 & a_2^5
    \end{bmatrix}
    \]
    is invertible.  The determinant of this matrix is given by $a_1^2a_2^2(a_2^3 - a_1^3)$, and hence we only require that $a_1\neq 0$ and $a_2\neq 0$, and that $a_2/a_1$ is not a third root of unity.  Setting $a_1 = 1$ and $a_2 = 3$ suffices, so valid evaluation vectors exist over the base field itself and in this example the lower bound of Proposition \ref{field_size} is sharp.

    However, if we take $h$ such that $\supp(h) = \{0,1,2,4,7,8,10\}$, then $\supp(\widehat{h}) = \{2,8\}$ and the relevant determinant is $a_1^2a_2^2(a_2^6 - a_1^6)$.  This polynomial is identically zero on all pairs $(a_1,a_2)\in\mathbb{F}_7^2$, and thus we must pass to a finite extension (the lower bound of Proposition \ref{field_size} is therefore necessary but not sufficient).  We must choose $a_1,a_2\in\mathbb{F}_{7^r}^\times$ such that  $(a_2/a_1)^6\neq 1$.  The kernel of the map $x\mapsto x^{6}$ on $\mathbb{F}_{7^r}^\times$ is exactly $\mathbb{F}_{7}^\times$, so choosing $a_1 = 1$ and $a_2\in \mathbb{F}_{7^2}^\times \setminus \mathbb{F}_{7}^\times$ suffices.  Thus we can recover $v_2$ and $v_{8}$ with $M|\supp(\widehat{h})| = 6$ evaluations, three of which must be done over the quadratic extension $\mathbb{F}_{7^2}$.  
\end{example}

\section{Modular Polynomial Codes} \label{mpc}

\subsection{A General Polynomial Code Scheme}

For the rest of the paper we let $\mathbb{F}$ denote a finite field containing a {primitive} $M^{th}$ root of unity $\zeta$.  A user wishes to multiply two block matrices $A$ and $B$ over $\mathbb{F}$, defined by
\begin{equation}\label{AB}
A 
= \begin{bmatrix}
    A_0 \\ \vdots \\ A_{K-1}
\end{bmatrix}
=
\begin{bmatrix}
A_{k,m}
\end{bmatrix}_{\substack{0\leq k\leq K-1\\0\leq m\leq M-1 }},\quad
B 
= \begin{bmatrix}
    B_0 & \cdots & B_{L-1}
\end{bmatrix}
=
\begin{bmatrix}
B_{m,\ell}
\end{bmatrix}_{\substack{0\leq m\leq M-1\\ 0\leq \ell \leq L-1}}
\end{equation}
The matrices $A$ and $B$ are selected independently and uniformly at random from the respective spaces of all such block matrices.  We suppose that the individual matrices $A_{k,m}$ are all of the same size, and similarly for the $B_{m,\ell}$.  We also assume the products $A_{k,m}B_{m',\ell}$ are all well-defined and therefore of the same size.  We define
\begin{equation}\label{figi}
f_I := \sum_{k = 0}^{K-1}\sum_{m = 0}^{M-1} A_{k,m}x^{m+kM}
\quad\text{and}\quad
g_I:= \sum_{\ell = 0}^{L-1}\sum_{m = 0}^{M-1}B_{m,\ell} x^{M-1-m + \ell KM}
\end{equation}

Let $T$ be a positive integer.  We wish to augment our polynomials $f_I$ and $g_I$ with randomness to protect against worker $T$-collusion.  To do this, first note that the highest degree monomial of $f_Ig_I$ whose coefficient is one of the desired submatrix multiplications $A_kB_\ell$ is $x^{KML-1}$.  Our random polynomials will therefore take the form
\begin{equation}\label{frgr}
f_R := \sum_{t = 0}^{T-1}R_tx^{KML+\alpha_t}\quad \text{and}\quad g_R := \sum_{t = 0}^{T-1}S_tx^{KML+\beta_t}
\end{equation}
for some integers $0\leq\alpha_0<\cdots<\alpha_{T-1}$ and $0\leq\beta_0<\cdots<\beta_{T-1}$ and some i.i.d.\ random matrices $R_t$ and $S_t$ whose precise nature will be made clear in the following subsection.  We then define our polynomials $f$ and $g$ to be
\begin{equation}\label{f_and_g}
f := f_I + f_R,\quad 
g := g_I + g_R,\quad\text{and}\quad h := fg.
\end{equation}
In what follows, we use the notation $\alpha = (\alpha_0,\ldots,\alpha_{T-1})$ and $\beta = (\beta_0,\ldots,\beta_{T-1})$.  The following result relates the product of $A$ and $B$ as matrices with the product of $f$ and $g$ as polynomials.

\begin{thm}\label{decodability}
Let $A$ and $B$ be as in \eqref{AB}.  Then for any $\alpha,\beta$, the coefficient of $x^{M-1+kM+\ell KM}$ of the polynomial $h = fg$ is the $(k,\ell)^{th}$ block entry $A_kB_\ell$ of the product $AB$.
\end{thm}
\begin{proof}
Computing the product of $f$ and $g$ yields
\begin{align*}
h = fg &=  f_Ig_I + f_Ig_R + f_Rg_I + f_Rg_R \\
&= \sum_{k = 0}^{K -1} \sum_{\ell = 0}^{L-1} \sum_{m = 0}^{M-1}\sum_{m' = 0}^{M-1} A_{k,m}B_{m',\ell}x^{M-1-m'+m+kM+\ell KM} + (\text{terms of degree $\geq KML$}).
\end{align*}
As $m$ and $m'$ range between $0$ and $M-1$, the quantity $M-1-m'+m$ ranges between $0$ and $2M-2$.  For $m$ and $m'$ in this range, it is clear that $M-1-m'+m = M-1$ if and only if $m = m'$.  In this case, for a fixed $k$ and $\ell$ the coefficient is then $\sum_{m = 0}^{M-1}A_{k,m}B_{m,\ell} = A_kB_\ell$ as desired.
\end{proof}

Crucially, the index $i$ of every coefficient of $h$ which is a block entry of the product $AB$ satisfies $i\equiv M-1\Mod{M}$.  Thus such coefficients are in $\supp(\widehat{h})$ and the user loses no information about these block entries of $AB$ by passing to the mod-$M$ transform $\widehat{h} = \widehat{fg}$.  This observation and Theorem \ref{decodability} motivate the following definition, which is based on the notion of a polynomial code as defined in \cite{d2020gasp}.

\begin{definition}\label{mpc_defn}
    A \emph{Modular Polynomial (MP) code} $\mathcal{M} = \mathcal{M}(K,M,L,T,P,\alpha,\beta,\mathbb{K}/\mathbb{F},\zeta,a)$ consists of the following data:
    \begin{itemize}
        \item positive integers $K$, $M$, $L$, $T$, and $P$
        \item vectors $\alpha = (\alpha_0,\ldots,\alpha_{T-1})\in\mathbb{Z}_{\geq0}^T$ and $\beta = (\beta_0,\ldots,\beta_{T-1})\in\mathbb{Z}_{\geq0}^T$, satisfying
        \[\alpha_0<\cdots <\alpha_{T-1}\quad \text{and}\quad \beta_0<\cdots<\beta_{T-1}
        \]
        \item a primitive $M^{th}$ root of unity $\zeta\in\mathbb{F}$
        \item a finite extension $\mathbb{K}/\mathbb{F}$ of finite fields
        \item an evaluation vector $a = (a_1,\ldots,a_P)\in\mathbb{K}^P$.
    \end{itemize}
    For any MP code $\mathcal{M}$, we have the associated \emph{recovery threshold} $N:=MP$.
\end{definition}

 A user wishes to multiply, securely and in a distributed manner, two matrices $A$ and $B$ as in \eqref{AB}, which are $K\times M$ and $M\times L$ block matrices defined over some finite field $\mathbb{F}$.  To do this using a modular polynomial code, they proceed according to the following steps:
\begin{enumerate}
    \item Construct polynomials $f$ and $g$ as in \eqref{f_and_g}, where the random matrices $R_t$ and $S_t$ are all chosen uniformly and independently from the set of matrices over $\mathbb{K}$ of the appropriate size.  
    \item Define $h:=fg$ and $P := |\supp(\widehat{h})|$ (hence $P$ cannot be chosen independently).
    \item Compute the $N=MP$ pairs of evaluations $f(\zeta^ma_p)$ and $g(\zeta^ma_p)$, {so that each $n$ for $1\leq n\leq N$ corresponds to a pair $(m,p)$ with $0\leq m\leq M-1$ and $1\leq p\leq P$.}
    \item Send the $n^{th}$ pair $f(\zeta^ma_p), g(\zeta^ma_p)$ to the $n^{th}$ worker node for $1\leq n\leq N$, who responds with \[
    h(\zeta^m a_p) = f(\zeta^ma_p)\cdot g(\zeta^ma_p)
    \]
    \item For each $p$, compute the evaluation
    \[
    \widehat{h}(a_p) = \frac{1}{M}\sum_{m = 0}^{M-1} \zeta^mh(\zeta^m a_p)
    \]
    \item Interpolate $\widehat{h}$ using the $P$ evaluations $\widehat{h}(a_p)$ to recover all block entries of $AB$.
\end{enumerate}

For now, our system model assumes that all worker nodes do indeed respond with the value of $h(\zeta^ma_p)$, that is, there are no stragglers.  We lift this assumption in Section \ref{robustness_section}.

At this point it is worth discussing the similarities and differences between the construction of MP codes and the codes appearing in \cite{root_of_unity}, which are a generalization of those of \cite{inner_product} from the inner product partition to the grid partition.  The fact that all $M^{th}$ roots of unity sum to zero is also used in \cite[Section 3]{root_of_unity}, in some sense to also reduce the total number of worker nodes.  However, the construction of \cite{root_of_unity} requires the evaluation points themselves to all be $N^{th}$ roots of unity, where $N$ is their recovery threshold.  Thus the nature of the evaluation points is different in both schemes, though those of \cite{root_of_unity} are more explicit.  Lastly, the encoding polynomials of \cite{root_of_unity} are Laurent polynomials whose nature appears fundamentally different than the $f$ and $g$ of \eqref{f_and_g}.  We revisit this comparison with experimental data in Section \ref{comparison}.

\begin{remark}
    In the present work, {we often emphasize the need to pass to a finite extension $\mathbb{K}/\mathbb{F}$ to find a valid evaluation vector $a$.  This is essentially equivalent to assuming the base field $\mathbb{F}$ is sufficiently large, as is done in \cite{oliver, d2020gasp, d2021degree} and other work on SDMM where the user is tasked with interpolating a polynomial with `gaps' in the degrees, or indeed any polynomial of sufficiently large degree.  However, the current approach allows us to not put any restrictions on the field $\mathbb{F}$ of definition of $A$ and $B$ (other than the existence of a primitive $M^{th}$ root of unity), and therefore emphasizes polynomial code constructions for SDMM do in fact provide schemes when $A$ and $B$ are defined over relatively small fields as well.}  We note that finding the smallest field over which various polynomial codes exist is an open and interesting question.
\end{remark}

\subsection{Decodability and $T$-Security of Modular Polynomial Codes}

\begin{definition}\label{decodability_t_security_defn}
    Suppose we are given a modular polynomial code $\mathcal{M}$ with parameters as in Definition \ref{mpc_defn} and a pair of matrices $A$ and $B$ as in \eqref{AB}.
    \begin{enumerate}
        \item We say that $\mathcal{M}$  is \emph{decodable} if
        \[{
        H(AB\ |\ \{(\widehat{h}(a_1),\ldots,\widehat{h}(a_{P}))\}) = 0.}
        \]
        \item We say that $\mathcal{M}$ is \emph{$T$-secure} if
        \[
        I\left(A,B\ ;\{(f(\zeta^{m_t}a_{p_t}),g(\zeta^{m_t}a_{p_t}))\ |\ t = 0,\ldots,T-1\}\right) = 0
        \]
        for all size-$T$ subsets $\{\zeta^{m_0}a_{p_0},\ldots,\zeta^{m_{T-1}}a_{p_{T-1}}\}$ of the set of evaluation points.
    \end{enumerate}
\end{definition}

It will be convenient for us to now define $T\times N$ matrices $\Sigma_A$ and $\Sigma_B$ over $\mathbb{K}$, which depend on $\alpha$, $\beta$, the primitive $M^{th}$ root $\zeta$, and the evaluation vector $a=(a_1,\ldots,a_P)\in\mathbb{K}^P$.  We set
\begin{align}
    \Sigma_A &:= \begin{bmatrix}
        \Sigma_{A,p}
    \end{bmatrix}_{1\leq p\leq P}\in \mathbb{K}^{T\times N},
    \quad\text{where}\quad
    \Sigma_{A,p} := \begin{bmatrix}
        (\zeta^m a_p)^{\alpha_t}
    \end{bmatrix}_{\substack{ 0\leq t\leq T-1 \\ 0\leq m\leq M-1 }} \in \mathbb{K}^{T\times M} \label{sigma_a} \\
    \Sigma_B &:= \begin{bmatrix}
        \Sigma_{B,p}
    \end{bmatrix}_{1\leq p\leq P}\in \mathbb{K}^{T\times N},
    \quad\text{where}\quad
    \Sigma_{B,p} := \begin{bmatrix}
        (\zeta^m a_p)^{\beta_t}
    \end{bmatrix}_{\substack{ 0\leq t\leq T-1 \\ 0\leq m\leq M-1 }}  \in \mathbb{K}^{T\times M} \label{sigma_b}
\end{align}
The matrices $\Sigma_A$ and $\Sigma_B$ encode the security properties of the MP code, as the following theorem illustrates.

\begin{thm}\label{mpc_decodability_security}
    Let $\mathcal{M} = \mathcal{M}(K,M,L,T,P,\alpha,\beta,\mathbb{K}/\mathbb{F},\zeta,a)$ be a modular polynomial code.  Then:
    \begin{enumerate}
        \item If the matrix $GV(a,\supp(\widehat{h}))$ is invertible, then $\mathcal{M}$ is decodable.
        \item If $a_p\neq 0$ for all $p$ and the matrices $\Sigma_A$ and $\Sigma_B$ have the MDS property, then $\mathcal{M}$ is $T$-secure.
    \end{enumerate}
\end{thm}
\begin{proof}
    The proof is essentially identical to the proof of decodability and $T$-security for many polynomial codes in the literature. {Intuitively,
    \begin{enumerate}
        \item If the matrix $GV(a,\supp(\widehat{h}))$ is invertible then the polynomial $\widehat{h}(x)$ can be interpolated from the evaluations $\widehat{h}(a_p)$ for $p=1,\ldots, P$. Therefore, the user can recover all of the coefficients of $\widehat{h}(x)$ and therefore all block entries $A_kB_\ell$ of the product $AB$.
        \item If $\Sigma_A$ has the MDS property, then any $T$-tuple of matrices $f(\zeta^{m_0}a_{p_0}),\ldots,f(\zeta^{m_{T-1}}a_{p_{T-1}})$ is uniform random on the set of all $T$-tuples of matrices of the appropriate size, and is independent of $A$.  Therefore, any set of $T$ workers observe only noise and can not decode any information about $A$. The same argument works for $B$ and it follows that $\mathcal{M}$ is $T$-secure.
    \end{enumerate}
    For more details} see, for example, the analogous proof for GASP codes in \cite[Theorem 1 and Appendix A]{d2020gasp}.
\end{proof}

If $M = 1$ then we are in the situation of most previous works on SDMM, and for any $\alpha$, $\beta$, showing the existence of an evaluation vector $a$ in some finite extension that guarantees decodability and $T$-security is largely a formality.  One simply has to observe that the two conditions in the previous theorem impose finitely many polynomial conditions on such an evaluation vector, {and appeal to Lemma \ref{sz_non_vanishing}}.

However, as soon as $M>1$ and $T>1$, there exist $\alpha$, $\beta$ for which no valid evaluation vector exists, because the property of $\Sigma_A$ and $\Sigma_B$ being MDS forces non-trivial conditions on $\alpha$, $\beta$ independent of the evaluation vector $a$.  We demonstrate this phenomenon with the following example.

\begin{example}\label{careful}
    Suppose that $T = 2$ and $M>1$.   The condition of $T$-security demands that the $2\times N$ matrix $\Sigma_A$ have the MDS property.  In particular, every $2\times M$ submatrix $\Sigma_{A,p}$ must also have the MDS property.  The determinant of any $2\times 2$ submatrix of $\Sigma_{A,p}$ is of the form
    \[
    a_p^{\alpha_0+\alpha_1}
    \begin{vmatrix}
    \zeta^{m_0\alpha_0} & \zeta^{m_1\alpha_0} \\
    \zeta^{m_0\alpha_1} & \zeta^{m_1\alpha_1}
    \end{vmatrix}
    =
    a_p^{\alpha_0+\alpha_1}(\zeta^{m_0\alpha_0 + m_1\alpha_1} - \zeta^{m_0\alpha_1 + m_1\alpha_0})
    \]
    where $0\leq m_0\neq m_1\leq M-1$.  This quantity is non-zero if and only if $a_p\neq 0$ and
    \[
    (m_0-m_1)(\alpha_0-\alpha_1)\not\equiv 0 \Mod{M}.
    \]
    Now as $m_0-m_1$ takes on every non-zero value mod $M$ as $m_0$ and $m_1$ vary, we conclude that $\Sigma_{A,p}$ has the MDS property if and only if $\alpha_0-\alpha_1$ {is relatively prime to $M$}.  Clearly the same property must hold for $\beta_0$ and $\beta_1$.
\end{example}

\begin{definition}
    An increasing sequence of integers $\alpha_0,\ldots,\alpha_{T-1}$ is in \emph{arithmetic progression} if $\alpha_t = d_0 + tD$ for all $t$, for some $d_0\in\mathbb{Z}$ and some $D\in\mathbb{Z}_{>0}$.  We refer to $D$ as the \emph{common difference} of the progression.
\end{definition}

\begin{thm}\label{polynomials}
    Suppose that $\alpha$ and $\beta$ are each in arithmetic progression with common differences $D$ and $E$, respectively.  Consider the elements $a_0,\ldots,a_{P-1}$ as variables over $\mathbb{F}$.  If $\gcd(D,M) = \gcd(E,M) = 1$, then the determinant of every $T\times T$ submatrix of $\Sigma_A$ and $\Sigma_B$ is non-zero as an element of the polynomial ring $\mathbb{F}[a_0,\ldots,a_{P-1}]$.
\end{thm}
\begin{proof}
    Let $X_0,\ldots,X_{T-1}$ be variables over $\mathbb{F}$, and define
    \[
    V_{\alpha}(X_0,\ldots,X_{T-1}) = \det\begin{bmatrix}
        X_t^{\alpha_t}
    \end{bmatrix}\in\mathbb{F}[X_0,\ldots,X_{T-1}].
    \]
    If $\alpha_t = d_0 + tD$, then a simple calculation using the classical expression for the Vandermonde determinant shows that
    \[
    V_{\alpha}(X_0,\ldots,X_{T-1}) = \prod_{t = 0}^{T-1}X_t^{d_0}\cdot \prod_{0\leq s < t \leq T-1}(X_t^D - X_s^D).
    \]
    Now the determinant of any $T\times T$ submatrix of $\Sigma_A$ is obtained by substituting $X_t = \zeta^m a_p$ for some $m$ and $p$ into the above expression for $V_{\alpha}(X_0,\ldots,X_{T-1})$. The resulting polynomial in the variables $a_0,\ldots,a_{P-1}$ will be non-zero exactly when $(\zeta^{m_1}a_{p_1})^D\neq (\zeta^{m_2}a_{p_2})^D$ as monomials for any pair of distinct substitutions $X_t = \zeta^{m_1}a_{p_1}$ and $X_s = \zeta^{m_2}a_{p_2}$.

    If $a_{p_1}\neq a_{p_2}$ then the condition $(\zeta^{m_1}a_{p_1})^D\neq (\zeta^{m_2}a_{p_2})^D$ holds unconditionally.  If $a_{p_1}=a_{p_2}$ then we must have $m_1\neq m_2$, in which case $(\zeta^{m_1}a_{p_1})^D\neq (\zeta^{m_2}a_{p_2})^D$ if and only if $(m_1 - m_2)D\not\equiv 0\Mod{M}$.  Since $D$ is assumed to be {relatively prime to} $M$, this latter condition is satisfied.  The same argument works for $\Sigma_B$ and the parameter $E$, which completes the proof of the theorem.
\end{proof}

\begin{corollary}\label{finally}
Suppose that $\alpha$ and $\beta$ are in arithmetic progression with common differences $D$ and $E$, respectively.  Assume that $\gcd(D,M) = \gcd(E,M) = 1$.  Then there exists an evaluation vector $a\in \mathbb{K}^P$ in some finite extension $\mathbb{K}/\mathbb{F}$, such that the resulting MP code is decodable and $T$-secure.
\end{corollary}
\begin{proof}
    It is well-known that the determinant of a generalized Vandermonde matrix is non-zero as a polynomial in its entries.  Thus the decodability condition forces a single polynomial condition of the form $\phi(a)\neq 0$ on any potential evaluation vector $a$.  By Theorem \ref{polynomials}, each $T\times T$ sub-determinant of $\Sigma_A$ and $\Sigma_B$ forces another such polynomial condition.  {One can now use Lemma \ref{sz_non_vanishing} to demonstrate the existence of the desired field extension and evaluation vector}.
\end{proof}

The results of this section have provided us with a robust set of modular polynomial codes which guarantee decodability and $T$-security.  In particular, the set of admissible $\alpha$ and $\beta$ includes that of $\alpha_t = \beta_t = t$, {which is the same choice of $\alpha_t$ and $\beta_t$ as in} the GASP$_{\mathrm{big}}$ polynomial code from \cite{d2020gasp}.  {Thus this specific class of MP codes can be thought of as a generalization of the GASP$_\mathrm{big}$ codes to the grid partition.}

\subsection{An Example: $K = L = 2$, $M = 3$, and $T = 3$}\label{mpc_exmaple}

At this point it is worth pausing to construct an explicit example which may demystify some of the abstraction of the previous two subsections.  We set $K = L = 2$, $M = 3$, and $T = 3$.  We suppose our matrices $A$ and $B$ are defined over the field $\mathbb{F}_{13}$, which contains a primitive third root of unity $\zeta = 3$.  For simplicity we choose $\alpha_t = \beta_t = t$ for $0\leq t \leq 2$.

The polynomials $f$ and $g$ are defined to be
\begin{align*}
    f &= A_{0,0} + A_{0,1}x + A_{0,2}x^2 + A_{1,0}x^3 + A_{1,1}x^4 + A_{1,2}x^5 + R_0x^{12} + R_1x^{13} + R_2x^{14} \\
    g &= B_{2,0} + B_{1,0}x + B_{0,0}x^2 + B_{2,1}x^6 + B_{1,1}x^7 + B_{0,1}x^8 + S_0x^{12} + S_1x^{13} + S_2x^{14}.
\end{align*}
Define $C = AB$, so that $C$ has four blocks $C_{i,j}$, where $0\leq i,j\leq 1$.  The polynomial $\widehat{h} = \widehat{fg}$ is then easily computed to be
\begin{align*}
\widehat{h} &= C_{0,0}x^2 + C_{1,0}x^5  + C_{0,1}x^8 + C_{1,1}x^{11} + D_{14}x^{14} + D_{17}x^{17} + D_{21}x^{20} + D_{26}x^{26}
\end{align*}
where the matrices $D_i$ are linear combinations of products of the blocks of $A$ and $B$ and the random matrices $R_t$ and $S_t$, whose precise description is irrelevant.  Hence $P = |\supp(\widehat{h})| = 8$ and this MP code has a recovery threshold of $N = MP = 24$.

To complete the scheme description it remains to find an evaluation vector $a=(a_1,\ldots,a_8)\in \mathbb{F}_{13^r}^8$ for which the decodability and $T$-security properties of Theorem \ref{mpc_decodability_security} are satisfied.  By Proposition \ref{field_size} taking $r = 1$ is insufficient, since $13 < 25 = MP + 1$.  However, over the finite extension $\mathbb{F}_{13^2}$ admissible evaluation vectors abound.  Let us represent this quadratic extension as
\[
\mathbb{F}_{13^2} = \mathbb{F}_{13}[x]/(x^2 + 12x + 2)
\]
and let $b$ be a root of this minimal polynomial ($b$ happens to be a multiplicative generator of $\mathbb{F}_{13^2}$).   One can check that, for example, choosing the evaluation vector to be
\[
a = (b, 11+b, 11+ 12b, 2 + 10b, 6 + 12b, 2 + 5b, 3 + 7b, 12 + 10b) \in \mathbb{F}_{13^2}^8
\]
results in an MP code which is decodable and $T$-secure, as do many other choices of evaluation vector over this quadratic extension.

\subsection{Explicit Recovery Threshold of Modular Polynomial Codes}\label{mpc_rate_section}

Recall from Corollary \ref{finally} that $\alpha$ and $\beta$ which are in arithmetic progression with common differences relatively prime to $M$ guarantee $T$-security of the corresponding modular polynomial code.  To compute the rate of such modular polynomial codes explicitly, we restrict the set of possible codes further, by assuming that
\begin{equation}\label{assume}
\alpha_t = \beta_t = tD,\quad D\leq M,\quad\text{and}\quad \gcd(D,M) = 1.
\end{equation}
That is, we set $d_0 = e_0 = 0$, assume the common difference $D$ is equal for both arithmetic progressions, and that $D\leq M$.  In practice, it seems that all modular polynomial codes with decent rate satisfy these assumptions.

The following theorem allows one to express the recovery threshold $N = MP$ of MP codes as an explicit function of the system parameters.

\begin{thm}\label{mpc_rate}
    Let $\mathcal{M}$ be a modular polynomial code with parameters as in Definition \ref{mpc_defn}.  Assume $\alpha_t = \beta_t = tD$ where $\gcd(D,M) = 1$ and $D\leq M$.  Define quantities $\ell_0$ and $t_0$ by
    \[
        \ell_0 := \min\left(1 +\left\lfloor  \frac{(T-1)D - 1}{KM}\right\rfloor, L-1 \right) \quad\text{and}\quad t_0 := \left(\left\lceil \frac{-KM + M}{D} \right\rceil + T - 1 \right)^+
    \]
    where $(x)^+ = \max(x,0)$.  Then the recovery threshold $N$ of this MP code is given by $N = MP$, where
    \begin{align*}
    P &= K(L + \ell_0) + (L-\ell_0 - 1)\min\left(\left\lfloor \frac{(T-1)D}{M} \right\rfloor +  1,K\right) + \left\lfloor \frac{(T-1)D}{M} \right\rfloor +  1 \\
    & + \left\lfloor \frac{(2T-2-t_0)D + 1}{DM}\right\rfloor + \delta
    \end{align*}
    and $\delta\in\{0,1\}$.
\end{thm}
\begin{proof}
    {The proof proceeds by writing $\supp(h)$ as a disjoint union of intervals, and counting the number of integers $i$ in each interval which satisfy $i\equiv M-1 \Mod{M}$.}  See Appendix \ref{mpc_rate_proof} {for a complete proof}.  
\end{proof}

We hypothesize that the expression for $P$ in Theorem \ref{mpc_rate} is a (not necessarily strictly) increasing function of $D$, which we have verified experimentally for a number of parameters but we cannot prove.  However, this experimental observation motivates considering only those MP codes for which $D = 1$ in Section \ref{comparison} where we compare these codes with others in the literature.

One can see from the proof of Theorem \ref{mpc_rate} in Appendix \ref{mpc_rate_proof} that the quantity $\delta$ is explicitly calculable, but its value depends on a somewhat unpredictable coincidence of divisibility.  One can argue heuristically that we will `usually' have $\delta = 0$.

{Let us conclude this section by observing that if we set $T = 0$ then one can verify directly that the recovery threshold for MP codes is $N = MP = KML$, which is the same number of submatrix multiplications necessary in the uncoded case.  Thus MP codes are optimal in this idealized scenario.  Moreover, when $T = 0$ we have $|\supp(h)| = KML + M - 1$, so the naive approach of simply interpolating $h$ to recover $AB$ leaves one with a provably suboptimal scheme in the absence of security.}





\section{Generalized GASP codes}\label{ggasp_section}

In this section we construct generalizations of the codes GASP$_r$, which were first defined in \cite{d2020gasp, d2021degree} for $M = 1$.  As we do not use the full machinery of MP codes, and as this is a straightforward generalization of results appearing in \cite{d2020gasp,d2021degree}, we keep our scheme description and recovery threshold results brief.

\subsection{Definition of Generalized GASP Codes}

In this subsection we focus on generalizations of the GASP codes first appearing in \cite{d2020gasp,d2021degree}, which have previously been constructed for the case of $M = 1$.  These codes are parameterized by a positive integer $r$ satisfying $1\leq r\leq \min(K,T)$.  Our generalizations will also be parameterized by a positive integer $r$ satisfying $1\leq r\leq \min(KM,T)$.  When one sets $M = 1$ and obtains the setting of \cite{d2020gasp,d2021degree}, it is easy to see that the following definition simplifies to the original construction of GASP codes.

\begin{definition}\label{ggasp_defn}
    The \emph{generalized GASP code} $\mathcal{G} = \mathcal{G}(K,M,L,T,r,N,\mathbb{K}/\mathbb{F},a)$ consists of the following data:
    \begin{itemize}
        \item positive integers $K$, $M$, $L$, $T$, and $N$
        \item a positive integer $r\leq \min(KM,T)$
        \item a finite extension $\mathbb{K}/\mathbb{F}$ of finite fields
        \item an evaluation vector $a = (a_1,\ldots,a_N)\in\mathbb{K}^N$
    \end{itemize}
The above data defines $\alpha,\beta\in\mathbb{Z}_{\geq0}^T$ according to the following formulas:
\begin{align*}
\alpha_0,\ldots,\alpha_{T-1} &= \text{the first $T$ integers in the sequence } \bigcup_{u = 0}^{\infty}\  [uKM:uKM + r - 1]\\
\beta_t &= t \quad \text{for $t = 0,\ldots,T - 1$}.
\end{align*}
We call $N$ the \emph{recovery threshold} of $\mathcal{G}$.
\end{definition}

A user wishing to securely multiply two matrices $A$ and $B$ in a distributed fashion using a GGASP code observes the following familiar protocol.  The user constructs $f$ and $g$ as in \eqref{f_and_g}, where $R_t$ and $S_t$ are all chosen uniformly and independently at random from the set of matrices over $\mathbb{K}$ of the appropriate size.  They define $h := fg$ and $N := |\supp(h)|$, and compute the $N$ pairs of evaluations $f(a_n)$ and $g(a_n)$.  The user then sends the $n^{th}$ pair $f(a_n),g(a_n)$ to the $n^{th}$ worker node, who responds with $h(a_n) = f(a_n)g(a_n)$.  The user then decodes $AB$ from the $N$ evaluations $h(a_n)$ as in Theorem \ref{decodability}.

One defines information-theoretic conditions of decodability and $T$-security as before, and as in \cite{d2020gasp,d2021degree} one {uses Lemma \ref{sz_non_vanishing} to show} that for any parameters $K$, $M$, $L$, $T$, and $r$, there exists a finite extension $\mathbb{K}/\mathbb{F}$ and an evaluation vector $a\in \mathbb{K}^N$ such that $\mathcal{G}(K,M,L,T,r,N,\mathbb{K}/\mathbb{F},a)$ is decodable and $T$-secure.  We omit the details.

Note that GGASP codes do not fit the profile of a modular polynomial code, since the integers $\alpha_t$ are not in arithmetic progression and hence there is apparently no guarantee of decodability or $T$-security if one were to choose evaluation points of the form $\zeta^ma_p$ and attempt to interpolate the corresponding polynomial $\widehat{h}$.

\subsection{Recovery Threshold of Generalized GASP Codes}

To begin computing $N = |\supp(h)|$ for GGASP codes, note that we can alternatively describe the integers $\alpha_t$ by first writing
\[
T = U r + r_0,\quad U = \lfloor T/r \rfloor\quad\text{and}\quad 0\leq r_0 < r.
\]
Then since $r\leq KM$ we have
\[
\{\alpha_0,\ldots,\alpha_{T-1}\} = \bigcup_{u = 0}^{U - 1} [uKM : uKM + r - 1]\ \cup\  [UKM : UKM  + r_0 - 1].
\]
Due to the above `tail' term in the expression of the $\alpha_t$'s as a union of intervals, our explicit expression for $|\supp(h)|$ is slightly different depending on whether $r|T$ or $r\nmid T$.

\begin{thm}\label{ggasp_r_rate}
Let $\mathcal{G}$ be a generalized GASP code with parameters as in Definition \ref{ggasp_defn}.  Let $T = Ur + r_0$ where $U = \lfloor T/r\rfloor$ and $0 \leq r_0 < r$, and define quantities $\ell_0$, $S_\ell$, and $V$ by  
    \begin{align*}
    \ell_0 &:= \min\left(1 + \left\lfloor \frac{T-2}{KM}\right\rfloor,L\right) \\
    S_\ell & := \left\{
    \begin{array}{cl}
    M + r - 1 & \text{if } 0\leq \ell \leq L -1 \\
    \max(M,T) + r - 1 & \text{if } L\leq \ell \leq L + U - 2 \\
    T + r - 1 & \text{if } \ell = L + U - 1 \text{ and } r_0 = 0 \\
    \max(M + r_0, T + r) - 1 & \text{if } \ell = L + U - 1 \text{ and } r_0 > 0 \\
    0 & \text{if } \ell = L+U \text{ and } r_0 = 0 \\
    T + r_0 - 1 & \text{if } \ell = L + U \text{ and } r_0 > 0
    \end{array}
    \right. \\
    V &:= \left\{
    \begin{array}{cl}
    S_{L+U-1} + S_{L+U} & \text{if } r_0 = 0 \\
    \min(S_{L+U-1},KM) + S_{L+U} & \text{if } r_0 > 0
    \end{array}
    \right.
    \end{align*}
    Then the recovery threshold $N$ is given by
    \begin{align*}
    N = KML + \max(\ell_0KM + \min(S_{\ell_0},KM), KM + T - 1) + \sum_{\ell = \ell_0 + 1}^{L+U-2}\min(S_\ell, KM) + V.
    \end{align*}
\end{thm}
\begin{proof}
    See Appendix \ref{ggasp_proof}.
\end{proof}

For the outer production partition when $M = 1$, the authors of \cite{d2020gasp, d2021degree} find optimal or near-optimal values of $r$, in several cases.  We forego any attempt at analytically solving for the value of $r$ which minimizes the recovery threshold $N$ due to the complicated nature of the expression in Theorem \ref{ggasp_r_rate}.

\subsection{An Example: $K = L = 5$, $M = 2$, $T = 4$, $r = 2$}\label{exampleGGASP}

Let us conclude our discussion of GGASP codes with an explicit example.  For the parameters $K = L = 5$, $M = 2$, and $T = 4$, one can use Theorem \ref{ggasp_r_rate} to perform a simple search over all $r$ such that $1\leq r\leq \min(KM,T) = 4$ to see that the $r$ which minimizes $N$ is $r = 2$.  

For these parameters, the vectors $\alpha$ and $\beta$ are given by $\alpha = (0,1,10,11)$ and $\beta = (0,1,2,3)$, and hence the polynomials $f$ and $g$ are defined to be
\begin{align*}
    f &= A_{0,0} + A_{0,1}x + A_{1,0}x^2 + A_{1,1}x^3 + A_{2,0}x^4 + A_{2,1}x^5 \\
    &+ A_{3,0}x^6 + A_{3,1}x^7 + A_{4,0}x^8 + A_{4,1}x^9 \\
    &+ R_0x^{50} + R_1x^{51} + R_2x^{60} + R_3x^{61}\\
    g &= B_{1,0} + B_{0,0}x + B_{1,1}x^{10} + B_{0,1}x^{11} + B_{1,2}x^{20} + B_{0,2}x^{21} \\
    &+ B_{1,3}x^{30} + B_{0,3}x^{31} + B_{1,4}x^{40} + B_{0,4}x^{41} \\
    &+ S_0x^{50} + S_1x^{51} + S_2x^{52} + S_3x^{53}.
\end{align*}
Then $h = fg$ has $\deg(h) = 114$ but the recovery threshold is $|\supp(h)| = N = 82$ as predicted by Theorem \ref{ggasp_r_rate}.  By apparent coincidence, the MP code with the same parameters $K = L = 5$, $M = 2$, and $T = 4$ with $\alpha = \beta = (0,1,2,3)$ also achieves $N = MP = 82$, though with a possibly different set of evaluation points.  In contrast, for these same parameters the SDMM scheme of \cite{root_of_unity} requires $N = 84$ worker nodes, and that of \cite{oliver} requires $N = 97$.  In the next section, we present a more systematic comparison between these four SDMM schemes.

\section{Comparison with Other Polynomial Codes} \label{comparison}

\subsection{Rate}\label{comp_rate}
To compare the various polynomial codes in the literature against each other we make the following definition of rate, which normalizes the recovery threshold to be less than one and allows for comparison when we vary the system parameters.

\begin{definition}\label{rate}
The \emph{rate} of a polynomial code with matrix partition parameters $K$, $M$, and $L$ and recovery threshold $N$ is defined to be $\mathcal{R} = KML/N$.
\end{definition}

In terms of computation, the rate  measures the ratio of submatrix multiplications of the form $A_{k,m}B_{m',\ell}$ performed by the standard textbook matrix multiplication algorithm, namely $KML$, to the number of such multiplications needed when computing the product in a distributed fashion and accounting for $T$-collusion.  The connection between rate, recovery threshold, and various measures of communication efficiency is elaborated on in \cite{rafael_note}.

\subsection{Experimental Setup}

In the following subsections we present experimental results which compare MP and GGASP codes to current results in the literature, specifically, the codes of \cite{root_of_unity} which we refer to as Root of Unity (ROU) codes, and the codes of \cite{oliver} which we refer to as BGK codes (a concatenation of the authors' last initials).  All such codes have the same partitioning of the matrices $A$ and $B$ into $K\times M$ and $M\times L$ block matrices, respectively, which enables a fair comparison across all such strategies.  Deserving of mention are the SGPD codes of \cite{flex}, which are also directly comparable with the current work but which are largely outperformed by the codes of \cite{oliver}.  Thus to keep our experimental results easily digestible we omit comparison with \cite{flex}.

For every figure in this section, we choose for simplicity the modular polynomial code with $\alpha_t = \beta_t = t$ for $t = 0,\ldots,T-1$.  The GGASP codes depend on the parameter $r$, and we remove this dependency by plotting the rate for the GGASP code with $r$ that minimizes the recovery threshold $N$, for given values of $K$, $M$, and $L$.  Such an optimal $r$ was computed by brute force.

For the ROU codes we use the formula for $N$ given in \cite[Theorem 1]{root_of_unity}.  To compute the rate of the BGK codes, we compute $N$ as the number of non-zero coefficients of the product of the polynomials $P(x)$ and $Q(x)$ defined in \cite[Section 3.1]{oliver}, which are analogous to our polynomials $f$ and $g$.

{The Discrete Fourier Transform codes of \cite{inner_product} provably outperform both MP and GGASP codes for the inner product partition where $K = L = 1$.  For the outer product partition where $M = 1$, both MP and GGASP codes reduce to various versions of the codes GASP$_r$ of \cite{d2020gasp, d2021degree}, whose relative performance has been analyzed in these references.  Thus for all experiments we restrict ourselves to $K,L>1$ and $M>1$, to compare the relative performance of all polynomial codes when the matrices are partitioned non-trivially in both the horizontal and vertical directions.}

\subsection{Rate as a Function of $T$, for Large $M$}

\begin{figure}[h!]
    \centering
    \includegraphics[width=0.48\textwidth]{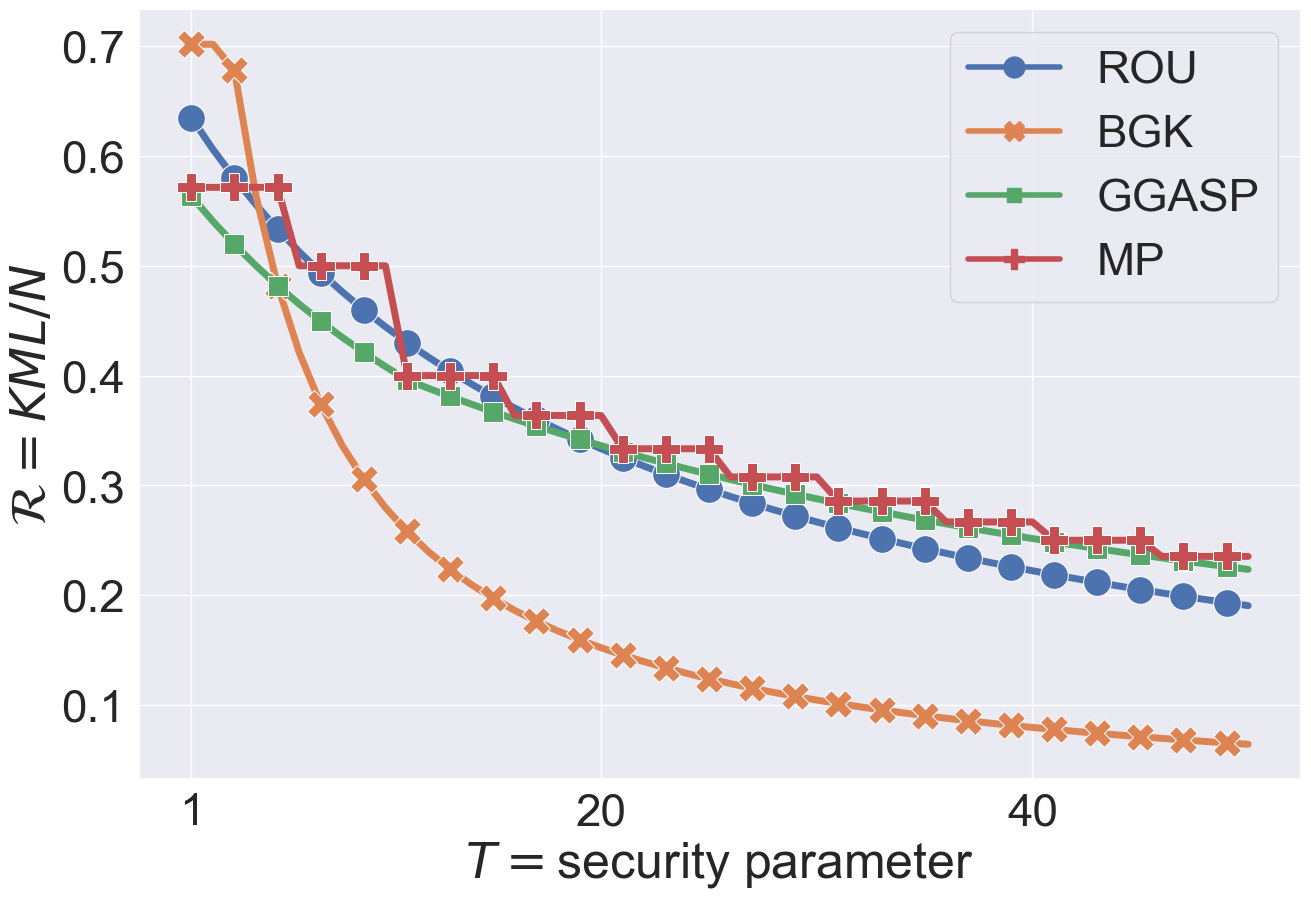}
    \includegraphics[width=0.48\textwidth]{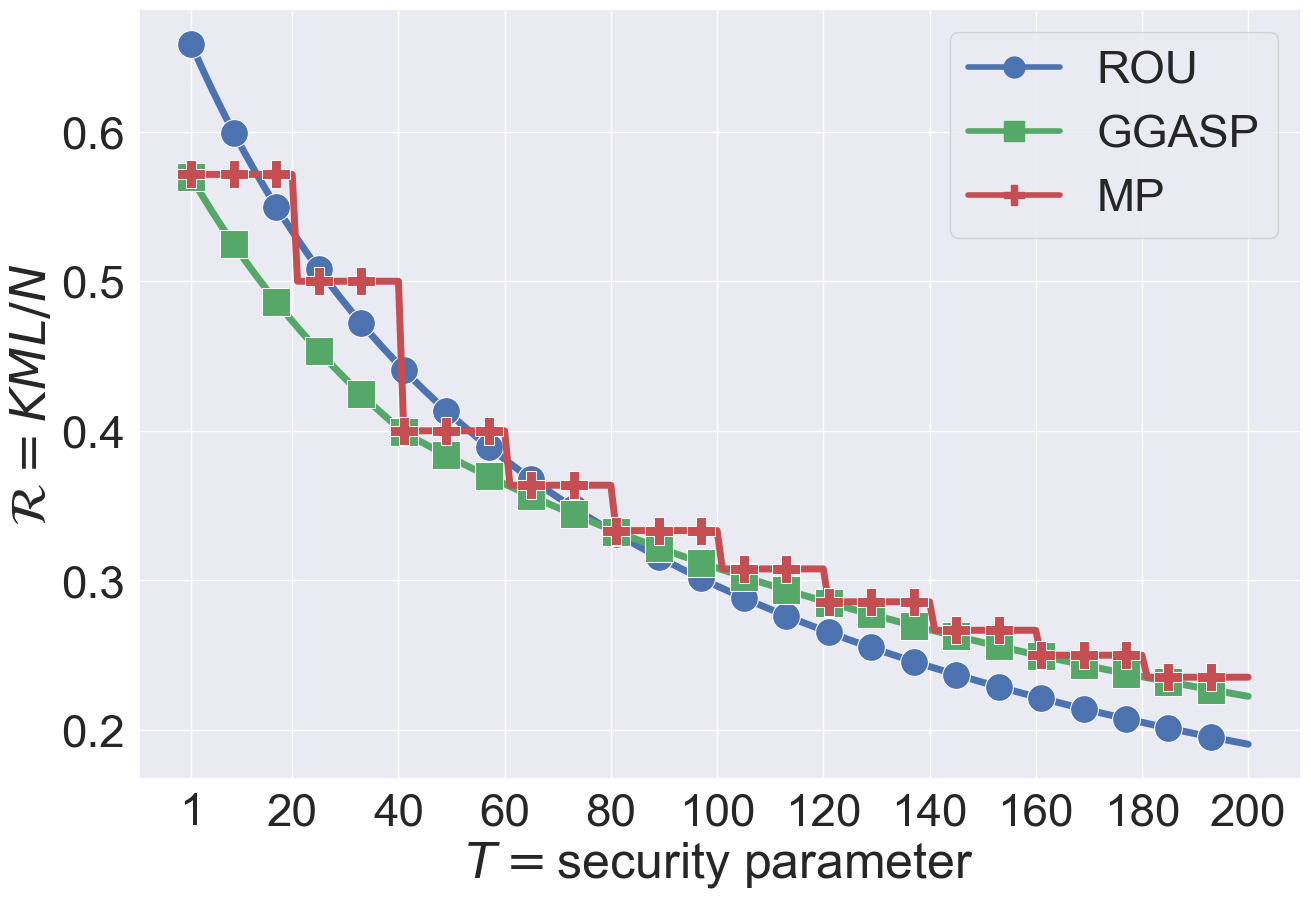}
    \caption{Rates of various polynomial codes, with $K = L = 2$, $M = 10$ (left), and $K = L = 2$, $M = 40$ (right), plotted as a function of the security parameter $T$.}
    \label{fig:inner}
\end{figure}

In Figure \ref{fig:inner} we plot the rate of various SDMM schemes as a function of $T$, in the regime where $M$ is somewhat larger than $K$ and $L$.   {We omit the performance of BGK codes from the right plot, as the rate values were close to zero for nearly all $T$.  In both plots we observe that for small $T$, the polynomial code with better performance alternates between ROU and MP.  However, for sufficiently large $T$ MP codes exhibit the best performance.}  

{Heuristically, the good performance of MP codes in this regime where $M\gg K,L$ can be explained by observing that for these parameters, a smaller proportion of the coefficient $i$ of the product polynomial $h$ satisfy $i\equiv M-1\Mod{M}$. The mod-$M$ transform is designed to decode these coefficients with as few evaluations as possible, in the process removing the `noise' terms given by the other coefficients.  In the asymptotic regime where $T\gg0 $, we see that MP codes outperform GGASP codes since for very large $T$ both codes assign $\alpha_t = \beta_t = t$, but MP codes use fewer worker nodes as they only interpolate $\widehat{h}$ instead of the full polynomial $h$.}

\subsection{Rate as a Function of $T$, for Large $K$ and $L$}

\begin{figure}[h!]
    \centering
    \includegraphics[width=0.48\textwidth]{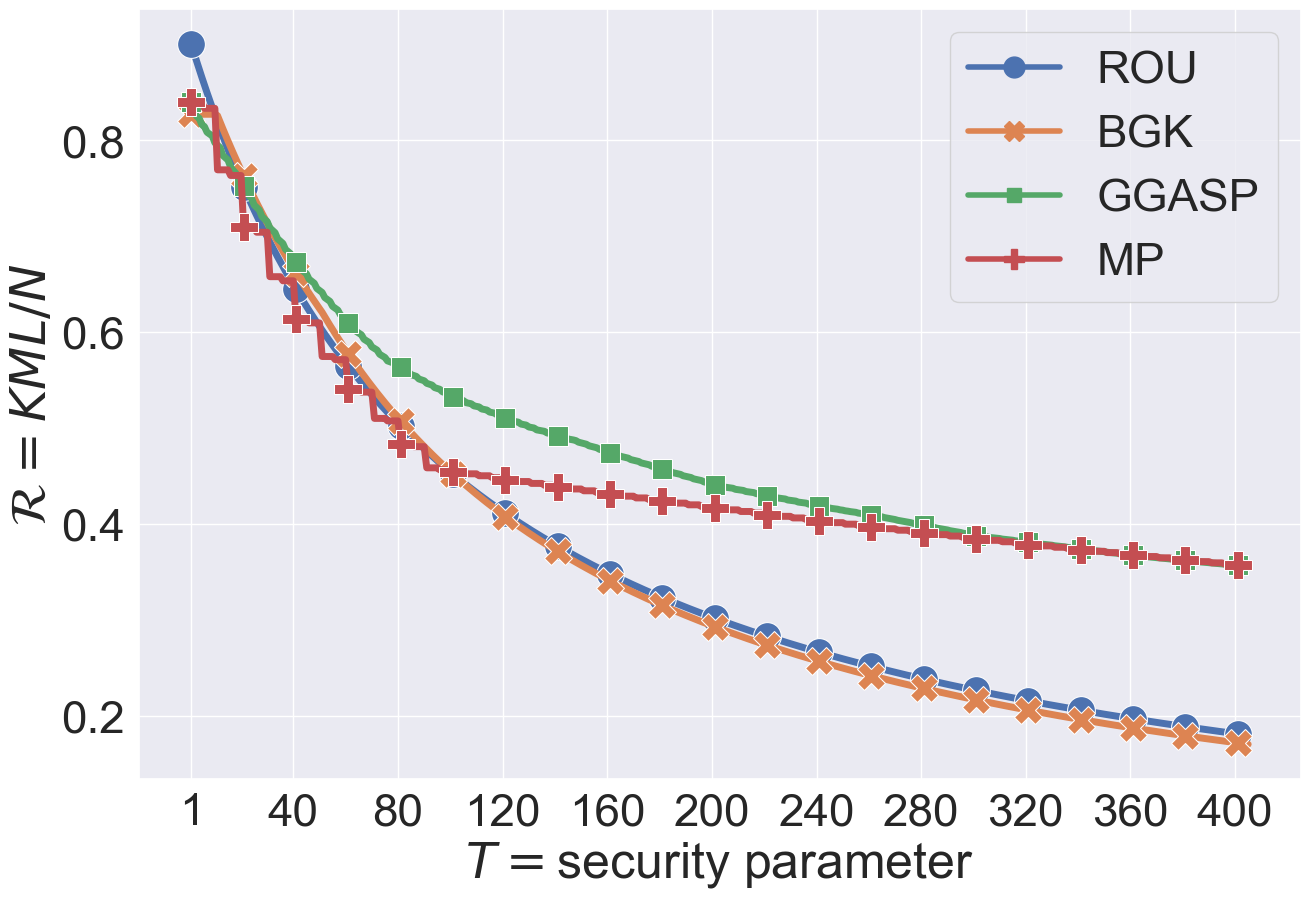}
    \includegraphics[width=0.48\textwidth]{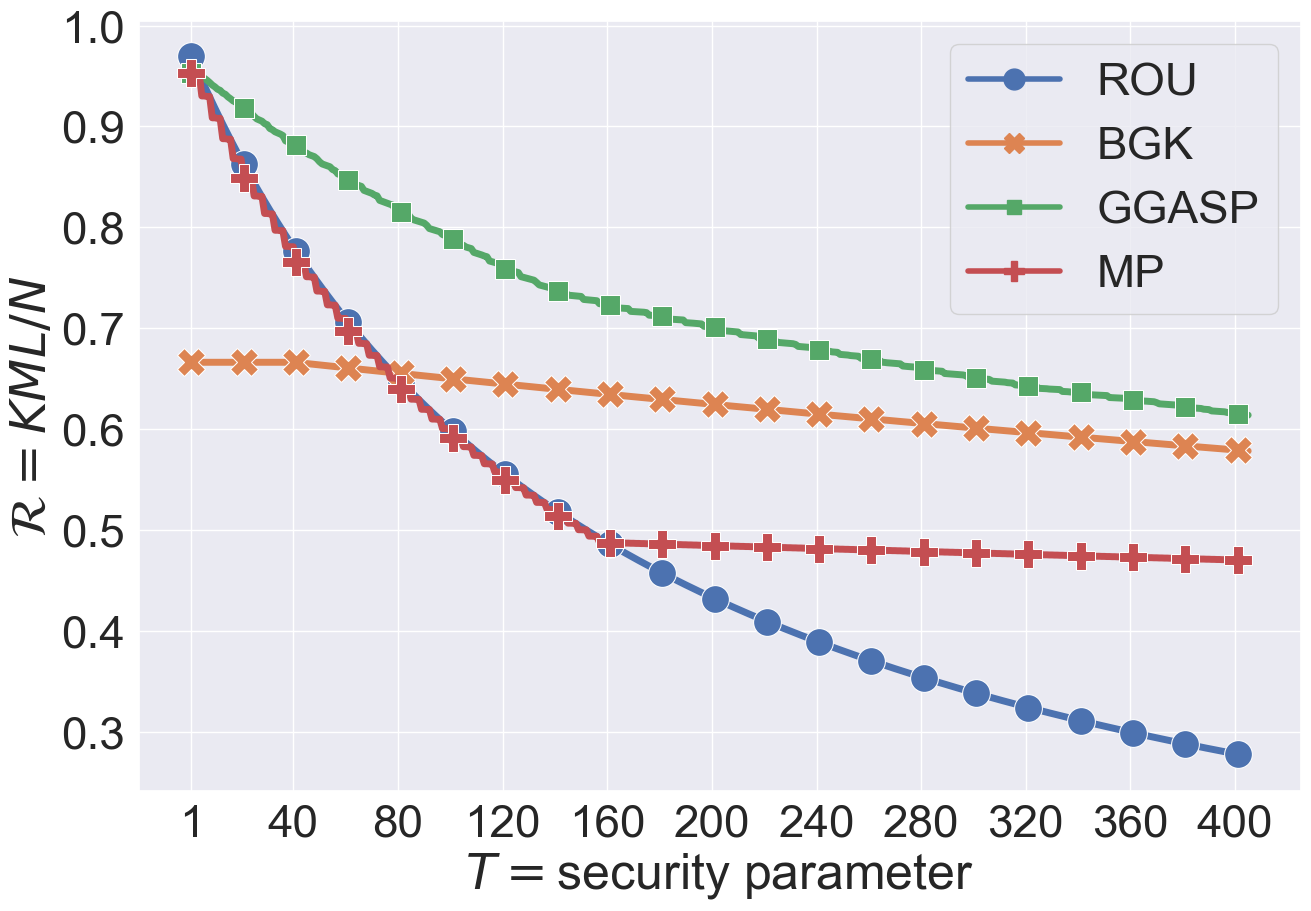}
    \caption{Rates of various polynomial codes, with $K = L = 10$, $M = 10$ (left), and $K = L = 40$, $M = 4$ (right), plotted as a function of the security parameter $T$.}
    \label{fig:inner2}
\end{figure}

In Figure \ref{fig:inner2} we make analogous plots {when $K$ and $L$ are larger relative to $M$, and extend the plots for larger $T$ to better understand these asymptotic regimes.  In both plots we notice that GGASP exhibits the overall best performance.  In the left plot when $K = L = M = 10$ we see that ROU, BGK, and MP codes exhibit comparable performance for small $T$, but as $T$ increases the MP codes outperform ROU and BGK codes by a large margin, and for these parameters eventually catch up to and even surpass GGASP codes in rate for $T\gg 0$ (though by a miniscule margin).}  {In the right plot we investigate asymptotics in $K$ and $L$ by setting $K = L = 40$ and $M = 4$.  Again GGASP codes perform best, which is not surprising as this parameter range is similar to the outer product partition for which the original GASP codes \cite{d2020gasp,d2021degree} were designed.}  On the other hand, the performance of the MP codes is somewhat lackluster when {$K$ and $L$ are much bigger than $M$ and $T$ is relatively small}.

\subsection{Fixing the Number of Worker Nodes and Limiting Computational Power}

Lastly, we suppose a user has available a fixed number $N$ of worker nodes, and would like to employ the SDMM scheme which maximizes the rate.  For each of the four codes in question, we compute by brute force the matrix partition parameters $K$, $M$, and $L$ which minimize the resulting number $N$ of worker nodes, and therefore maximize the  rate $\mathcal{R}$.  We also add an additional degree of realism to this situation by restricting the possible values of $K$, $M$, and $L$ to be greater than some lower bound; this effectively limits the computational power of each worker node.  In Figure \ref{fig:N200} the  resulting rates are plotted as a function of $T$ for $N = 200$ and $K,L\geq 2$, $M\geq 4$ in the left plot, and $K,L\geq 2$, $M\geq 10$ in the right plot.

\begin{figure}[h!]
    \centering
    \includegraphics[width=0.48\textwidth]{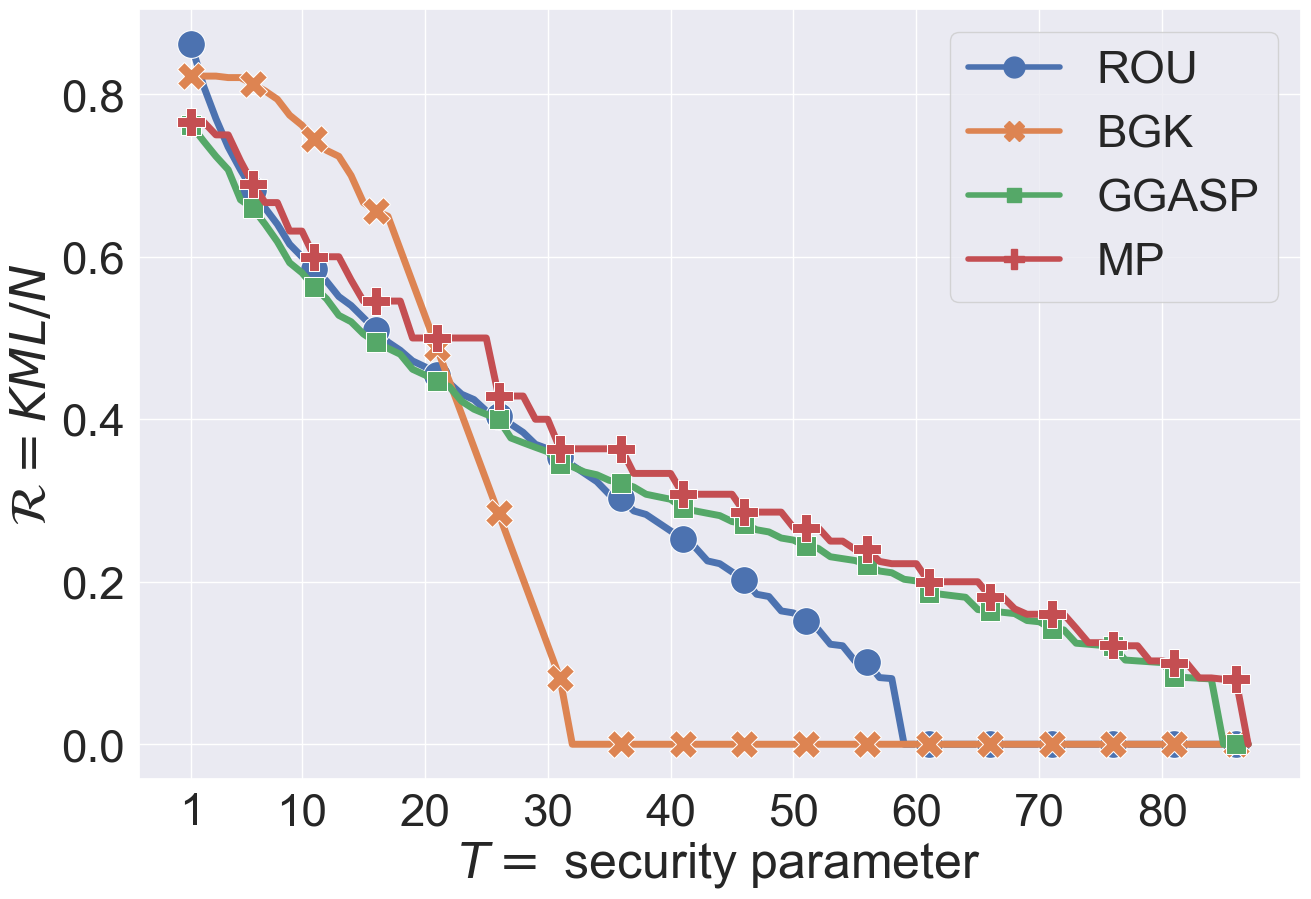}
    \includegraphics[width=0.48\textwidth]{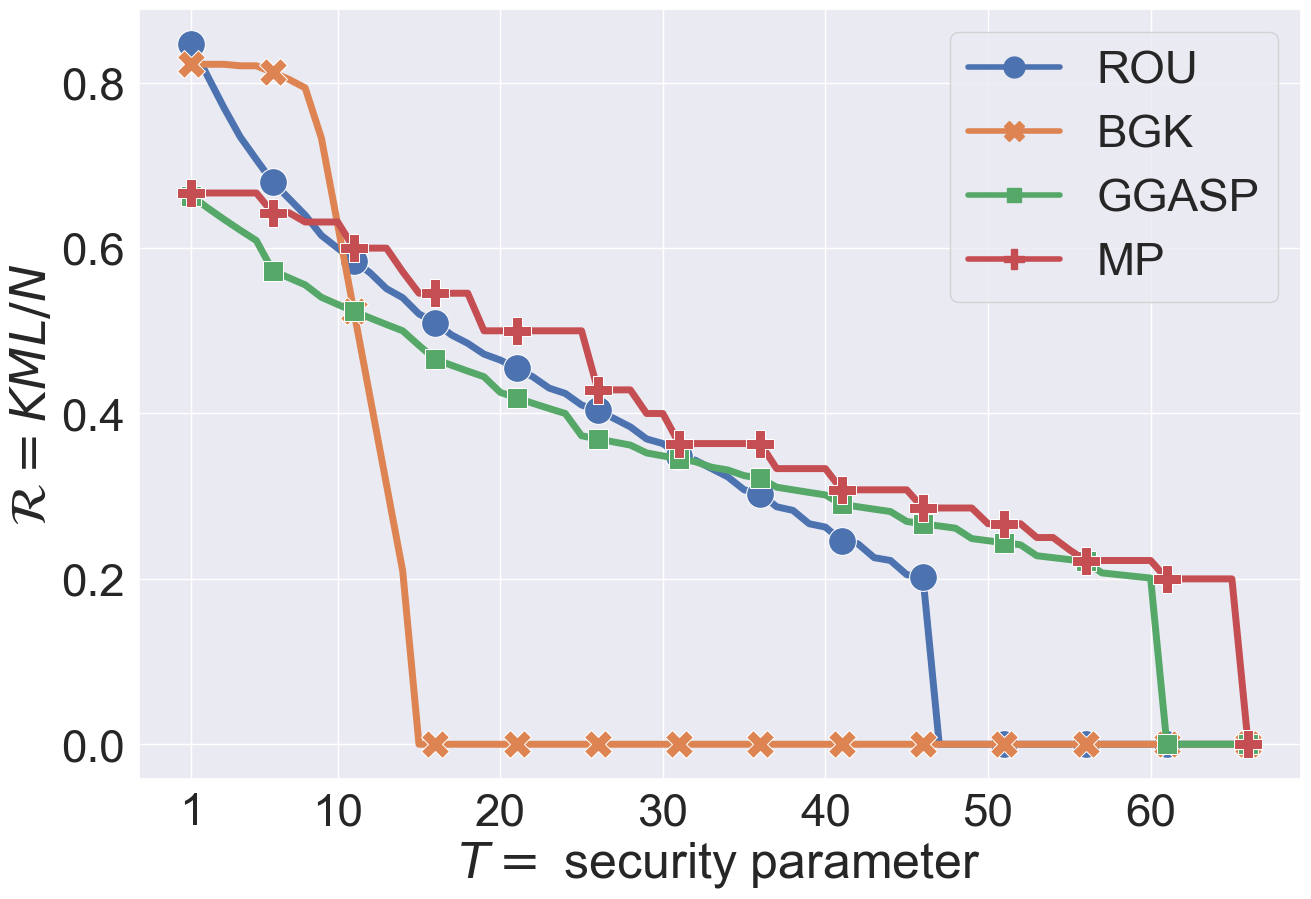}
    \caption{Rates of various polynomial codes with $N = 200$, $K,L\geq 2$, $M\geq 4$ (left),  and $N = 200$, $K,L\geq 2$, $M\geq 10$ (right), plotted as a function of the security parameter $T$.}
    \label{fig:N200}
\end{figure}

  We observe in Figure \ref{fig:N200} that both MP and GGASP codes show noticeable improvements over the current codes in the literature when $T$ is a sizeable fraction of $N$.  Given a value of $N$, both MP and GGASP codes provide non-zero rate schemes for much larger values of $T$ than ROU or BGK codes. 
 However, BGK codes maintain an edge in the low-$T$ regime.

\section{Robustness of MP and GGASP Codes Against Stragglers} \label{robustness_section}

We now shift our focus to the problem of \emph{robustness} against stragglers.  More specifically, we add more worker nodes to our system than the recovery thresholds $N$ of Definitions \ref{mpc_defn} and \ref{ggasp_defn} demand, in anticipation of some subset of $S$ worker nodes being unresponsive. The notion of recovery threshold is generalized below, and coincides with the definition found in \cite{straggler_fund, straggler_optimal, polycodes} when $T = 0$, and that of Definitions \ref{mpc_defn} and \ref{ggasp_defn} when $S = 0$.  We note here that since correcting a single error is equivalent to repairing two erasures, this recovery threshold alternatively allows for robustness against a subset of $\lfloor S/2\rfloor$ Byzantine worker nodes which return erroneous responses.

\subsection{MP Codes for Robustness}

We generalize the notion of MP codes to add robustness against stragglers as follows.  We let $P' = |\supp(\widehat{h})|$ and choose some $P\geq P'$.  With the rest of the parameters as before, the user employs $N = MP$ worker nodes by choosing an evaluation vector $a \in \mathbb{K}^P$ defined over some finite extension $\mathbb{K}/\mathbb{F}$.  The $N$ worker nodes are sent the evaluations $f(\zeta^ma_p),g(\zeta^ma_p)$ and respond with the products $h(\zeta^ma_p)$.  The user decodes by computing the evaluations $\widehat{h}(a_p)$ and interpolating $\widehat{h}$ as before.  The definition of $T$-security as in Definition \ref{decodability_t_security_defn} remains unchanged, but the definition of decodability is generalized to that of robustness as below.

For the rest of this section, we adopt the following notation.  We let $\mathcal{W}$ be the set of all $N$ worker nodes, $\mathcal{W}_{m,p}$ be the worker node corresponding to the evaluation point $\zeta^ma_p$, and we define $\mathcal{W}_p = \{\mathcal{W}_{0,p},\ldots,\mathcal{W}_{M-1,p}\}$ to be the $p^{th}$ \emph{hypernode}.

\begin{definition}\label{mp_robust_defn}
    Let $\mathcal{M}$ be a modular polynomial code with $P' = |\supp(\widehat{h})|$.  Let $P\geq P'$, suppose the user employs $N = MP$ worker nodes as described above, and let $\mathcal{S}\subseteq\mathcal{W}$ be a subset of worker nodes.  We say that $\mathcal{M}$ is \emph{robust against $\mathcal{S}$} if
    \[
    H(AB\ |\ \{h(\zeta^ma_p)\ |\ \mathcal{W}_{m,p}\in \mathcal{W}\setminus \mathcal{S}\}) = 0.
    \]
    If $\mathcal{M}$ is robust against $\mathcal{S}\subseteq\mathcal{W}$ for all subsets of size $S = |\mathcal{S}|$, we say that $\mathcal{M}$ is \emph{robust against $S$ stragglers}.  If $S$ is maximal such that $\mathcal{M}$ is robust against $S$ stragglers, then $N-S$ is the \emph{recovery threshold} of $\mathcal{M}$.
\end{definition}

Thus $\mathcal{M}$ is robust against $S$ stragglers if the responses from any $N-S$ worker nodes suffice to decode $AB$.  The recovery threshold is then the minimum number of worker nodes necessary such that the responses from any subset of worker nodes of that size suffice to decode $AB$.  When $S = 0$ the recovery threshold is equal to $N$, hence the above definition coincides with that of Definition \ref{mpc_defn} in that case.

\subsection{Robustness of MP Codes in the Absence of Security}

We proceed by first analyzing the case of $T = 0$ when there is no security constraint.  This makes the arguments more intuitive and highlights the milder assumptions needed on the base field and the evaluation vector when we remove the condition of $T$-security.  

For ease of exposition we blur the distinction in what follows between scalar and matrix coefficients.  Thus the following results actually hold on a per-entry basis in the involved matrices, but we will not make this explicit.

\begin{lem}\label{h_is_rs}
    Let $T = 0$, so that $f = f_I$ and $g = g_I$ as in \eqref{figi}.  Let $h = fg$ and $N \geq KML + M - 1$.  Then for any evaluation vector $a\in\mathbb{F}^N$ with distinct entries, the vector $(h(a_1),\ldots,h(a_N))\in\mathbb{F}^N$ is a {codeword} of a Reed-Solomon code with length $N$ and dimension $KML + M - 1$.
\end{lem}
\begin{proof}
    One checks easily that $h$ is an arbitrary polynomial of degree $KML + M - 2$, hence $|\supp(h)| = KML + M - 1$ and the result follows.
\end{proof}

\begin{lem}\label{hhat_is_rs}
    Let $T = 0$, so that $f = f_I$ and $g = g_I$ as in \eqref{figi}.  Let $h = fg$ and $P\geq KL + 1$.  Then for any evaluation vector $a\in\mathbb{F}^P$ with distinct entries such that $a_p\neq 0$ for all $p$ and $a_p^M\neq a_q^M$ for all distinct $p,q$, the vector $(\widehat{h}(a_1),\ldots,\widehat{h}(a_P))\in\mathbb{F}^P$ is a {codeword} of a Generalized Reed-Solomon code with length $P$ and dimension $KL$.
\end{lem}

\begin{proof}
    We can write $\widehat{h}(x) = x^{M-1}k(x^M)$ where $k$ is an arbitrary polynomial of degree $KL$.  Thus the evaluation vector of $\widehat{h}$ is given by
    \[
    (\widehat{h}(a_1),\ldots,\widehat{h}(a_P)) = (a_1^{M-1}k(a_1^M),\ldots, a_P^{M-1}k(a_P^M))
    \]
    from which the result follows immediately based on our assumptions on the evaluation vector.
\end{proof}

\begin{thm}\label{mp_recovery_threshold}
    Let $\mathcal{M}$ be a modular polynomial code with $T = 0$ and $P\geq KL + 1$, and suppose we have $N = MP$ worker nodes.  Then for any evaluation vector $a\in\mathbb{F}^P$ with $a_p\neq 0$ and $a_p^M\neq a_q^M$ for all distinct $p,q$, we have:
    \begin{enumerate}
        \item The recovery threshold of $\mathcal{M}$ is $KML + M - 1$.
        \item $\mathcal{M}$ is robust against any subset $\mathcal{S}\subseteq\mathcal{W}$ such that $\mathcal{W}\setminus\mathcal{S}$ contains $KL$ hypernodes.
    \end{enumerate}
\end{thm}
\begin{proof}
    The assumption $P \geq KL + 1$ guarantees that $N = MP > KML + M - 1$.  To prove part 1), We must show that any $KML + M - 1$ worker nodes suffice to decode $AB$.  The assumptions on the evaluation points $a_p$ guarantee that the $\zeta^ma_p$ are all distinct, hence by Lemma \ref{h_is_rs} the evaluation vector $(h(\zeta^ma_p))\in\mathbb{F}^N$ lives in a Reed-Solomon code of dimension $KML + M - 1$.  It follows that any $KML + M - 1$ worker nodes suffice to decode all evaluations $h(\zeta^m a_p)$, from which the user can compute $\widehat{h}(a_p)$ and decode $AB$.

    To prove part 2), note that from any $KL$ hypernodes we can decode $KL$ of the evaluations $\widehat{h}(a_p)$, which by Lemma \ref{hhat_is_rs} suffices to decode all of $\widehat{h}$ and therefore $AB$.
\end{proof}

Lemma \ref{hhat_is_rs} and Theorem \ref{mp_recovery_threshold} both place the assumption on our vector $a = (a_1,\ldots,a_P)$ that $a_p\neq 0$ for all $p$ and that $a_p^M\neq a_q^M$ for all distinct $p,q$.  To satisfy these conditions one can choose the $a_p$ to be any $P$ elements from a subgroup of $\mathbb{F}^\times$ of order relatively prime to $M$.  See the example in Subsection \ref{rob_sec_example}.

\subsection{Decoding Below the Recovery Threshold}

Theorem \ref{mp_recovery_threshold} shows that the recovery threshold of MP codes matches that of Entangled Polynomial Codes \cite{straggler_fund} and the codes of \cite{oliver}, but in some `best-case' scenarios we can recover with fewer worker nodes than $KML + M - 1$.  Differing best-case and worst-case recovery thresholds have also been observed for Discrete Fourier Transform codes, see \cite[Remark 5 and Remark 6]{inner_product}.  To measure the robustness of MP codes against subsets of stragglers of size $S > N - (KML + M - 1)$, we present the following definition.

\begin{definition}\label{p_robustness}
    Let $\mathcal{M}$ be a modular polynomial code with $P' = |\supp(\widehat{h})|$.  Let $P\geq P'$ and suppose the user employs $N = MP$ worker nodes as described above.  For any $S\leq N$, define
    \[
    p(S) := \frac{|\{\mathcal{S}\subseteq\mathcal{W}\ |\ \text{$|\mathcal{S}| = S$ and $\mathcal{M}$ is robust against $\mathcal{S}$}\}|}{\binom{N}{S}}.
    \]
\end{definition}

The function $p(S)$ thus measures the probability that an MP code is robust against a subset of size $S$ chosen uniformly at random.  Restating the results of the previous subsection in this new language, we have $p(S) = 1$ if $S \leq N - (KML + M - 1)$, and $p(N - KML) > 0$.  The following result describes the function $p(S)$ in the regime of interest.

\begin{thm}\label{prob_recovery}
    Let $\mathcal{M}$ be a modular polynomial code with $T = 0$ and $P\geq KL + 1$, and suppose we have $N = MP$ worker nodes.  Let $a\in \mathbb{F}^P$ with $a_p\neq 0$ and $a_p^M\neq a_q^M$ for all distinct $p,q$.

    Let $N - (KML + M - 1) < S \leq N - KML$.  Then:
    \[
    p(S) \geq \frac{\binom{P}{KL}\binom{N-KML}{S}}{\binom{N}{S}}.
    \]
\end{thm}
\begin{proof}
    Consider a set $\mathcal{S}$ of worker nodes of size $|\mathcal{S}| = S$, such that $\mathcal{W}\setminus \mathcal{S}$ contains $KL$ hypernodes.  There are $\binom{P}{KL}$ ways to choose such a set of hypernodes.  For each such selection, $\mathcal{S}$ is a subset of the remaining $N - KML$ worker nodes, hence there are $\binom{N - KML}{S}$ ways to assign worker nodes to elements of $\mathcal{S}$.  Clearly $\mathcal{M}$ is robust against any such $\mathcal{S}$, as we have $KL$ evaluations $\widehat{h}(a_p)$ which suffices to decode $AB$ by Lemma \ref{hhat_is_rs}.  There are $\binom{P}{KL}\binom{N-KML}{S}$ such subsets $\mathcal{S}$, hence the result.
\end{proof}

\begin{figure}[h!]
    \centering
    \includegraphics[width=0.48\textwidth]{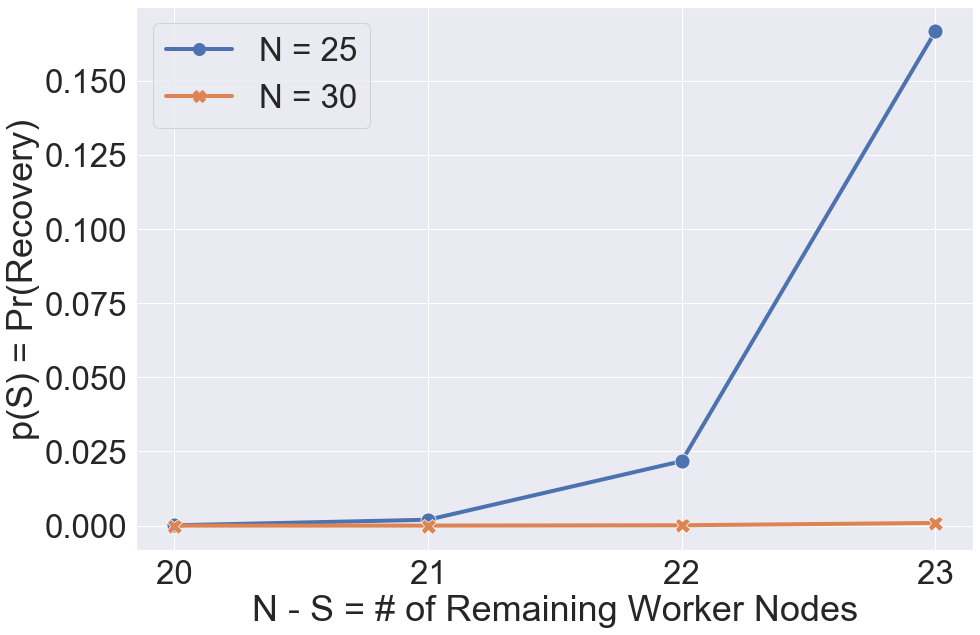}
    \includegraphics[width=0.48\textwidth]{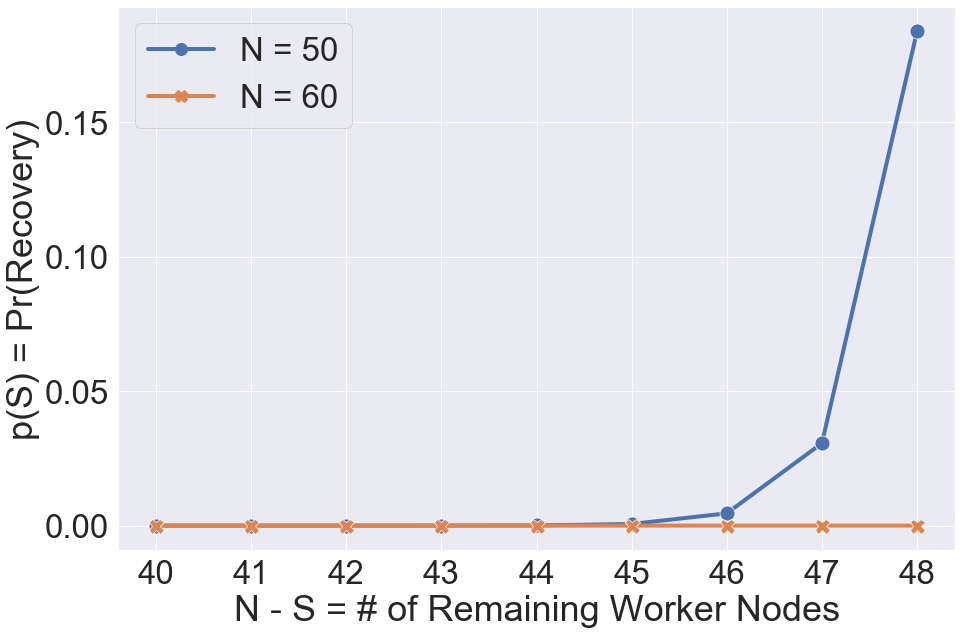}
    \caption{The lower bound on the function $p(S)$ of Theorem \ref{prob_recovery} as a function of $N-S$, for $K=L=2$ and $M=5$ (left), and $K=L=2$ and $M=10$ (right).  The different curves correspond to different numbers of total worker nodes.}
    \label{fig:rec_threshold}
\end{figure}

In Figure \ref{fig:rec_threshold} we plot the lower bound on $p(S)$ from Theorem \ref{prob_recovery} as a function of $N-S$, the number of non-straggler worker nodes.  We restrict to $KML \leq N-S < KML + M - 1$, as below this lower bound we have no method of decoding, and at this upper bound we can recover with any subset of workers of size $N - S$.  In both plots we take $P = KL + P'$ where $P' = 1,2$, and so $N = M(KL + P')$.  In both plots, we observe the greatest probability of successful decoding below the recovery threshold when $P' = 1$, that is, when we only guarantee robustness against any single straggler.  This may seem counterintuitive, but notice that as $N$ increases, we have a non-zero probability of successful decoding with fewer than $KML + M - 1$ worker nodes, with a smaller \emph{fraction} of the total number of workers.

For larger and thus more realistic system parameters $K$, $M$, and $L$, the lower bound on $p(S)$ coming from Theorem \ref{prob_recovery} is so small as to no longer have any practical importance.  Thus Theorem \ref{prob_recovery} should serve more as a theoretical result which demonstrates that successful decoding below the recovery threshold is even possible.

\subsection{Combining Security and Robustness}

The results of the above two subsections restrict to the case of $T = 0$, that is, no security constraint is placed on our system.  In this subsection we reintroduce the security constraint, that is, we let $T\geq 0$ be arbitrary,  and we let $f$ and $g$ be as in Equation \eqref{f_and_g} and define $h = fg$ as always.

\begin{lem}\label{nprime_mds}
    Let $N' = |\supp(h)|$ and let $N\geq N'$. 
 Then there exists a finite extension $\mathbb{K}/\mathbb{F}$ and an evaluation vector $a\in\mathbb{K}^N$ such that the vector $(h(a_1),\ldots,h(a_N))\in\mathbb{K}^N$ is a {codeword} of an MDS code with length $N$ and dimension $N'$.
 \end{lem}
\begin{proof}
For any finite extension $\mathbb{K}/\mathbb{F}$, we consider the code given by the image of the evaluation map
\begin{equation}\label{defn_eval}
\mathrm{ev}_a: \mathbb{K}[x]^{\supp(h)}\rightarrow \mathbb{K}^N,\quad \mathrm{ev}_a(\phi) = (\phi(a_1),\ldots,\phi(a_N))
\end{equation}
where $\mathbb{K}[x]^{\supp(h)} = \{\phi\in\mathbb{K}[x]\ |\ \supp(\phi)\subseteq\supp(h)\}$, which is clearly a vector space over $\mathbb{K}$ of dimension $N'$.  Thus the image of $\mathrm{ev}_a$ has length $N$ and dimension $N'$.  We have to show that we can choose $a$ so that this code is MDS.  The condition of this code being MDS is equivalent to the matrix  $GV(a,\supp(h))$ having the MDS property.  Thus if $a$ is chosen to avoid all zero sets of the $N'\times N'$ sub-determinants of this matrix, then the resulting code is MDS.  By {Proposition \ref{evaluation}}, such an $a$ always exists over some finite extension of $\mathbb{F}$.
\end{proof}

The notions of robustness and recovery threshold introduced in Definition \ref{mp_robust_defn} generalize in an obvious way to GGASP codes.  We have the following result:

\begin{thm}\label{ggasp_robust_t}
    Let $\mathcal{G}$ be a generalized GASP code, let $N' = |\supp(h)|$, and suppose we have $N\geq N'$ worker nodes.  Then there exists an evaluation vector $a\in\mathbb{K}^N$ in some finite extension  $\mathbb{K}/\mathbb{F}$ such that the recovery threshold of $\mathcal{G}$ is $N'$. 

    In particular, if $T = 0$ then the recovery threshold of $\mathcal{G}$ is $KML + M - 1$.
\end{thm}
\begin{proof}
    This follows from Lemma \ref{nprime_mds}.  One must choose $a$ to avoid the zero sets of the determinant polynomials arising from both the robustness and security constraints, but this amounts to a finite set of non-zero polynomials.
\end{proof}

Unlike MP codes, decoding GGASP codes below the recovery threshold is impossible even when $T = 0$.  The evaluation vector of $h$ belongs to a Reed-Solomon code, in which one cannot decode with fewer worker nodes than the dimension of the code.

\begin{lem}\label{pprime_mds}
    Let $P' = |\supp(\widehat{h})|$ and let $P\geq P'$.  Then there exists a finite extension $\mathbb{K}/\mathbb{F}$ and an evaluation vector $a\in\mathbb{K}^P$, such that the vector $(\widehat{h}(a_1),\ldots,\widehat{h}(a_P))\in\mathbb{K}^P$ is a {codeword} of an MDS code with length $P$ and dimension $P'$.
\end{lem}
\begin{proof}
    As in the proof of Lemma \ref{hhat_is_rs}, we can write $\widehat{h}(x) = x^{M-1}k(x^M)$ where $|\supp(k)| = P'$. We restrict to  evaluation vectors $a$ such that $a_p\neq 0$ for all $p$ and $a_p^M\neq a_q^M$ for all distinct $p,q$.   The proof now proceeds almost identically to that of Lemma \ref{nprime_mds} with $\supp(k)$ in place of $\supp(h)$, and $b_p = a_p^M$ serving as the variables in the relevant polynomial equations.
\end{proof}

\begin{thm}\label{mp_recovery_t}
    Let $\mathcal{M}$ be a modular polynomial code, let $P' = |\supp(\widehat{h})|$, and let $P\geq P'+1$.  Suppose we have $N = MP$ worker nodes.  Then there exists an evaluation vector $a\in\mathbb{K}^P$ in some finite extension $\mathbb{K}/\mathbb{F}$ such that $\mathcal{M}$ is $T$-secure and:
    \begin{enumerate}
        \item The recovery threshold of $\mathcal{M}$ is $\leq N-(P-P')$.
        \item If $T\gg0$ the recovery threshold of $\mathcal{M}$ is $N'=|\supp(h)|$.
        \item $\mathcal{M}$ is robust against any subset $\mathcal{S}\subseteq\mathcal{W}$ such that $\mathcal{W}\setminus\mathcal{S}$ contains any $P'$ hypernodes.
    \end{enumerate}
\end{thm}
\begin{proof}
    To see part 1), note that any subset of the worker nodes $\mathcal{W}$ which contains $\geq N - (P-P')$ worker nodes must contain at least $P'$ hypernodes.  Thus the responses from any such subset suffice to compute $P'$ of the values $\widehat{h}(a_p)$.  By Lemma \ref{pprime_mds} this suffices to interpolate $\widehat{h}$ and thus decode $AB$.  That the evaluation vector can be chosen to satisfy the $T$-security condition follows from Corollary \ref{finally}.

    For part 2), it is easy to see that for sufficiently large $T$ there are no `gaps'  in $\supp(h)$.  {To see this, note that for large $T$ one can express $\supp(h)$ as a union of translates of $[0:T-1]$, all of which must overlap for $T\gg0$.  We therefore have} $|\supp(h)| = \deg(h) + 1$.  Thus $(h(\zeta^ma_p))\in\mathbb{K}^N$ is a {codeword} of a Reed-Solomon code with length $N$ and dimension $N'$, and the result follows.  Our ability to choose an evalution vector that also satisfies the $T$-security condition is guaranteed by Corollary \ref{finally}.
    
    Part 3) follows immediately from Lemma \ref{pprime_mds}.
\end{proof}

Part 1) of Theorem \ref{mp_recovery_t} is likely pessimistic, as it does not take advantage of the erasure-correction properties of the obvious code containing the evaluation vector $(h(\zeta^ma_p))\in\mathbb{K}^N$.  {The central open question is whether} an evaluation vector $a\in\mathbb{K}^P$ can be found such that $(h(\zeta^ma_p))$ lives in an MDS code of length $N$ and dimension $|\supp(h)|$.  The main impediment to proving such a statement is that results  {such as Lemma \ref{sz_non_vanishing} only guarantee the existence of arbitrary evaluation vectors, whereas} for MP codes our evaluation vector must be of the form $(\zeta^ma_p)$.  Such a statement would immediately imply that the recovery threshold of MP codes subject to $T$-security is $N'$, as is the case with GGASP codes.  {For small parameters one can explicitly find such evaluation vectors by brute force as we demonstrate in the following subsection.  However, we lack sufficient empirical evidence or intuitive rationale to formally conjecture that an evaluation vector of the desired form always exists, so we leave this question open.}

\subsection{An Example: $K=L=2$ and $M=3$}\label{rob_sec_example}

We conclude this section with an example to concretely demonstrate the robustness of MP codes against stragglers.  We set $K = L = 2$ and $M = 3$, and we study the recovery threshold of the corresponding MP code for various values of $T$.  This case-by-case analysis shows that the recovery threshold of $N'$ is always achievable for these parameters and any value of $T$.

We suppose our matrices are defined over a field $\mathbb{F}$ such that $|\mathbb{F}|\equiv 1\Mod{3}$, and thus $\mathbb{F}$ contains a primitive $3^{rd}$ root of unity $\zeta$.  We choose $\alpha_t = \beta_t = t$, so that the polynomials $f$ and $g$ are defined to be
\begin{align*}
    f &= A_{0,0} + A_{0,1}x + A_{0,2}x^2 + A_{1,0}x^3 + A_{1,1}x^4 + A_{1,2}x^5 + \sum_{t = 0}^{T-1}R_tx^{12 + t} \\
    g &= B_{2,0} + B_{1,0}x + B_{0,0}x^2 + B_{2,1}x^6 + B_{1,1}x^7 + B_{0,1}x^8 + \sum_{t = 0}^{T-1}S_tx^{12 + t}
\end{align*}
and as always, we set $h = fg$ and $N' = |\supp(h)|$.  For all values of $T$, we set $P' = |\supp(\widehat{h})|$ and  $P = P' + 2$, so that the resulting MP code uses $N = MP = M(P'+2)$ worker nodes.  \newline

\subsubsection*{$T = 0$}  Here we ignore security concerns and use $N = MP = 18$ worker nodes.  We choose $|\mathbb{F}|$ large enough so that our vector $a$ is defined over $\mathbb{F}$.  We choose $\mathbb{F} = \mathbb{F}_{31}$ which contains a primitive $3^{rd}$ root $\zeta = 5$, and whose multiplicative group $\mathbb{F}^\times$ contains a subgroup of order $10>P$, namely, the subgroup generated by $\omega = 15\in\mathbb{F}^\times$.  We set $a_p = \omega^p$ for $p = 0,\ldots,5$, so that the  vector $(h(\zeta^m\omega^p))\in\mathbb{F}^{18}$ belongs to a Reed-Solomon code of length $N = 18$ and dimension $N' = 14 = KML + M - 1$, hence this latter number is the recovery threshold of the MP code, or, using the notation of Definition \ref{p_robustness}, we have $p(4) = 1$.  One can then use Theorem \ref{prob_recovery} to calculate explicitly that $p(5) = 0.0105$ and $p(6) = 0.0008$, hence these are the probabilities of recovery with $S = 5$ and $S = 6$ stragglers, respectively. \newline

\subsubsection*{$T = 1$} We have $\supp(h) = \{0,1,\ldots,20,24\}$ and so $N' = |\supp(h)| = 22$.  Additionally, we have $\supp(\widehat{h}) = \{2, 5, 8, 11, 14, 17, 20\}$ and so $P' = |\supp(\widehat{h})| = 6$.  We have $P = 8$ and use $N = MP = 24$ worker nodes.  For these parameters we can again use $\mathbb{F} = \mathbb{F}_{31}$, choose $\zeta = 5$ as our $3^{rd}$ root of unity, and set $a_p = \omega^p$ for $p = 0,\ldots,7$ where $\omega = 15\in\mathbb{F}^\times$.  Our evaluation points are thus all $24$ points $\zeta^m\omega^p$.  Now consider the image of the map
\begin{equation}\label{f_eval}
\mathrm{ev} :\mathbb{F}[x]^{\supp(h)}\rightarrow \mathbb{F}^N, \quad \mathrm{ev}(\phi) = (\phi(\zeta^m\omega^p))\in\mathbb{F}^N
\end{equation}
as in Equation \eqref{defn_eval}.  This is a linear code with generator matrix $GV((\zeta^ma_p),\supp(h))$, which one can check using CoCalc \cite{CoCalc} is an MDS code.  It follows that this MP code has recovery threshold $N' = 22$, so it is robust against any $S = 2$ stragglers and secure against any $T = 1$ worker nodes.  However, any subset of worker nodes of size $MP' = 18$ containing at least $P'= 6$ hypernodes suffices to decode $AB$. Note that the $T$-security condition of Definition \ref{decodability_t_security_defn} is clearly satisfied, since when $T = 1$ it reduces to no evaluation point being zero.  \newline

\subsubsection*{$T=2$} We have $\supp(h) = \{0,1,\ldots,21,24,25,26\}$ and so $N' = |\supp(h)| = 25$ and $P' = |\supp(\widehat{h})| = 8$.  We have $P = 10$ and therefore we use $MP = 30$ worker nodes.  We work over the field $\mathbb{F} = \mathbb{F}_{61}$.  Here we choose $\zeta = 47$, and the element $\omega = 8$ generates a subgroup of $\mathbb{F}^\times$ of order $20$.  Our evaluation points will be the $N = MP$ elements $\zeta^m\omega^p$ for $p \in \{0, 1, 2, 3, 4, 7, 8, 9, 12, 13\}$.  The condition of $T$-security reduces to $p_1\not\equiv p_2\Mod{20}$ for all distinct exponents $p_1$, $p_2$ in this set, which is clearly satisfied.  These exponents were found by a brute-force search using CoCalc \cite{CoCalc}.  The image of the evaluation map $\phi\mapsto (\phi(\zeta^m\omega^p))\in\mathbb{F}^N$ as in Equation \eqref{f_eval} is an MDS code with length $N = 30$ and dimension $N' = 25$.  The recovery threshold of this MP code is therefore $N' = 25$.  This is strictly smaller than the upper bound of $N - (P-P') = 28$ provided by Theorem \ref{mp_recovery_t}, and provides a single data point in favor of the loose conjecture that the recovery threshold of $N'$ is always achievable for MP codes.    \newline

\subsubsection*{$T\geq 3$}  For $T\geq 3$ one can see easily that there are no `gaps' in the degrees of $h$, hence $N' = |\supp(h)| = \deg(h) + 1$, a quantity which can be calculated exactly using Theorem \ref{ggasp_r_rate}.  We set $P' = |\supp(\widehat{h})|$, which itself can be calculated exactly using Theorem \ref{mpc_rate}, choose any $P\geq P' + 1$, and use $N = MP$ worker nodes.  Any evaluation vector $(h(\zeta^ma_p))\in\mathbb{F}^N$ whose evaluation points are distinct and satisfy the $T$-security condition in Theorem \ref{mpc_decodability_security} now belongs to a Reed-Solomon code of length $N$ and dimension $N'$, hence we can achieve a recovery threshold of $N'$ under the condition of $T$-security and robustness against any $S$ stragglers for any $T$ and $S$.

\section{Computational Complexity} \label{complexity}

In this section we briefly summarize the complexity of the required arithmetic operations performed by the user, for both MP and GGASP codes.  The necessary operations can be divided into those required for encoding, that is, evaluating the polynomials $f$ and $g$, and those required for decoding, that is, decoding $AB$ from the $N$ evaluations $h(a_n)$. We focus exclusively on counting the number of multiplications necessary for encoding and decoding and ignore additions.  

In what follows, we let $A_{k,m}\in\mathbb{F}^{a\times s}$ and $B_{m,l}\in\mathbb{F}^{s\times b}$, for all $k$, $m$, and $\ell$. {We assume throughout this section that the polynomial codes in question operate over a finite extension $\mathbb{K}/\mathbb{F}$, and multiplications are considered as multiplications in the field $\mathbb{K}$.}

\subsection{Complexity of Encoding}

For both MP and GGASP codes, the complexity of the encoding operation at the user is controlled by the cost of evaluating the matrix-valued polynomials $f$ and $g$ at $N$ elements of $\mathbb{K}$.

\begin{prop}\label{horner}
    Suppose that $f \in \mathbb{K}[x]$ is a polynomial and let $\supp(f) = \{i_1,\ldots,i_N\}$.  Define $\Delta_f = \max(i_1,i_2-i_1,\ldots,i_N-i_{N-1})$.  Then $f$ can be evaluated with $N -1 + O(\log(\Delta_f))$ multiplications in $\mathbb{K}$.
\end{prop}
\begin{proof}
    This is an obvious generalization of Horner's method.  Define $\Delta_1 = i_1$ and $\Delta_n = i_n - i_{n-1}$ for $n = 2,\ldots, N$.  We express $f$ as
    \[
        f = x^{\Delta_1}\left(a_{i_1} + x^{\Delta_2}(a_{i_2} + x^{\Delta_3}(a_{i_3} + \cdots + x^{\Delta_{N-1}}(a_{i_{N-1}} + a_{i_N}x^{\Delta_N} ))) \right)
    \]
    by which we can evaluate $f$ at some $b\in\mathbb{K}$ by computing $N-1$ multiplications and $N$ exponentiations.  Each exponentation $b\mapsto b^{\Delta_n}$ can be performed with $O(\log\Delta_n)$ multiplications using successive doubling, hence the result.
\end{proof}

\begin{thm}\label{enc_complexity}
Let $f$ and $g$ be as in Equation \eqref{f_and_g}.  Then $f$ and $g$ can be evaluated at any $N$ points of $\mathbb{K}$ with $NS$ multiplications in $\mathbb{K}$, where
\[
S = asKM + sbML + (as + sb)(T-1) + O(\log\Delta)
\]
and $\Delta = \max(\Delta_f, \Delta_g)$, where $\Delta_f$ and $\Delta_g$ are as in Proposition \ref{horner}.
\end{thm}
\begin{proof}
    This follows immediately from Proposition \ref{horner}.  One only needs to observe that $|\supp(f)| = KM + T$ and $|\supp(g)| = ML + T$, and the result follows.
\end{proof}

\subsection{Complexity of Decoding}

While the encoding complexity of MP and GGASP codes is essentially the same, the decoding complexity for GGASP codes depends on the complexity of interpolating $h$, while MP codes are tasked with the arguably simpler task of computing the evaluations $\widehat{h}(a_p)$ from the evaluations $h(\zeta^m a_p)$, and then interpolating $\widehat{h}$.

We assume that the user can pre-compute the values $\zeta^m$ as well as the inverse $1/M$ needed to compute $\widehat{h}$.  As these computations only need to be done once regardless of input matrix size or number of worker nodes, we ignore them in our calculation of decoding complexity.

\begin{thm}\label{dec_complexity}
    The decoding complexity for GGASP codes and MP codes are as follows:
    \begin{enumerate}
        \item For GGASP codes, decoding $AB$ by interpolating $h$ can be done with at most $O(abN\log N)$ multiplications in $\mathbb{K}$, where $N$ is the recovery threshold.
        \item For MP codes, decoding $AB$ by interpolating $\widehat{h}$ can be done with at most $O(ab\max(N, P\log P))$ multiplications in $\mathbb{K}$, where $N$ is the recovery threshold.
    \end{enumerate}
\end{thm}
\begin{proof}
    The proof for GGASP codes is the same as that of \cite[Proposition 7]{oliver}, which cites \cite{poly_complexity} for the complexity of interpolating a general polynomial.
    
    For MP codes, computing $\widehat{h}(a_p) = \frac{1}{M}\sum_{m=0}^{M-1}\zeta^mh(\zeta^m a_p)$ when the value of $h(\zeta^ma_p)$ is already known requires $M-1$ multiplications to compute each of the values $\zeta^mh(\zeta^ma_p)$ for $m\neq 0$, and one multiplication to divide the sum by $M$.  Thus each of the $P$ evaluation points $a_p$ requires $M$ multiplications, for a total of $N = MP$ multiplications to compute $\widehat{h}(a_p)$ for all $p$.  Repeating this process for all $ab$ entries of the blocks of $AB$ gives us $abN$ multiplications.  Since $|\supp(\widehat{h})| = P$, interpolating $\widehat{h}$ requires at most $O(P\log P)$ multiplications in $\mathbb{K}$ for each entry of $AB$, as is noted in \cite{poly_complexity}.  Thus $O(abP\log P)$ multiplications are needed to recover $AB$ from the evaluations $\widehat{h}(a_p)$.  This completes the proof of the theorem.
\end{proof}

As $\max(N,P\log P)\leq N\log N$, Theorem \ref{dec_complexity} shows that MP codes offer better decoding complexity than polynomial codes which require one to interpolate an entire polynomial whose support is equal in size to the number of worker nodes.  {Of course, a detailed, fair comparison with other polynomial codes would require all codes in question to be defined over the same field $\mathbb{K}$.}

\section{Conclusion and Future Work} \label{conclusion}

In this paper, we have introduced two new types of polynomial codes for Secure and Robust Distributed Matrix Multiplication, namely \emph{modular polynomial (MP) codes} and \emph{generalized GASP (GGASP) codes}.  Explicit conditions for decodability and $T$-security were studied and expressions for the recovery thresholds of such codes were given.  These codes were shown experimentally to compare favorably to the existing codes in the literature, for the grid partition. In terms of robustness against stragglers, both MP and GGASP codes achieve the state-of-the-art recovery threshold in the absence of security.  However, MP codes also allow the user to decode with fewer worker nodes than the recovery threshold would normally allow, depending on the set of worker nodes that fails.   We also computed encoding and decoding complexities for both families of codes, and showed that MP codes have lower decoding complexity than other codes which use the same number of worker nodes.

Given that the MP codes can occasionally decode below the recovery threshold, it would be interesting to test these codes against probabilistic system models, wherein one would measure the probability of successful decoding given failure probabilities for each worker node.  If the failures of worker nodes within the same hypernode are highly correlated, then one would expect MP codes to perform particularly well.

Many problems in Coding Theory require decoding coefficients of polynomials, or more generally functions on algebraic curves, from evaluations of such functions.  It is our hope that the technique of partial polynomial interpolation and the $\mathrm{mod}$-$M$ transform introduced in this paper have applicability outside the current framework to other such problems. 

The schemes given in the present work can be generalized to finding the product of more than two matrices, for example, by decoding the coefficients of a product $fgh$ of three polynomials corresponding in an obvious way to the product $ABC$ of three matrices. However, as the worker nodes have no apparent way of removing the noise terms inherent in the product $fg$ before multiplying it with $h$, such obvious generalizations seem doomed to have rate less than one even in the absence of both security and robustness constraints.  Thus constructing rate-efficient schemes for the multiplication of three matrices, which avoid computing the product sequentially using two rounds of communication, is an interesting avenue for future research.

{\small
\printbibliography

@article{poly_complexity,
title = {Fast modular transforms},
journal = {Journal of Computer and System Sciences},
volume = {8},
number = {3},
pages = {366-386},
year = {1974},
issn = {0022-0000},
doi = {https://doi.org/10.1016/S0022-0000(74)80029-2},
url = {https://www.sciencedirect.com/science/article/pii/S0022000074800292},
author = {A. Borodin and R. Moenck}
}

@misc{duursma,
      title={Parity-Checked Strassen Algorithm}, 
      author={Hsin-Po Wang and Iwan Duursma},
      year={2022},
      eprint={2011.15082},
      archivePrefix={arXiv},
      primaryClass={cs.IT}
}

@ARTICLE{jafar_batch2,
  author={Jia, Zhuqing and Jafar, Syed Ali},
  journal={IEEE Transactions on Information Theory}, 
  title={On the Capacity of Secure Distributed Batch Matrix Multiplication}, 
  year={2021},
  volume={67},
  number={11},
  pages={7420-7437},
  doi={10.1109/TIT.2021.3112952}}

@ARTICLE{jafar_cross,
  author={Jia, Zhuqing and Sun, Hua and Jafar, Syed Ali},
  journal={IEEE Transactions on Information Theory}, 
  title={Cross Subspace Alignment and the Asymptotic Capacity of  $X$ -Secure  $T$ -Private Information Retrieval}, 
  year={2019},
  volume={65},
  number={9},
  pages={5783-5798},
  doi={10.1109/TIT.2019.2916079}}

@ARTICLE{jafar_xsecure_tprivate,
  author={Jia, Zhuqing and Jafar, Syed Ali},
  journal={IEEE Transactions on Information Theory}, 
  title={X-Secure T-Private Information Retrieval From MDS Coded Storage With Byzantine and Unresponsive Servers}, 
  year={2020},
  volume={66},
  number={12},
  pages={7427-7438},
  doi={10.1109/TIT.2020.3013152}}

@misc{hermitian,
      title={HerA Scheme: Secure Distributed Matrix Multiplication via Hermitian Codes}, 
      author={Roberto A. Machado and Welington Santos and Gretchen L. Matthews},
      year={2023},
      eprint={2303.16366},
      archivePrefix={arXiv},
      primaryClass={cs.IT}
}

@misc{okko4,
      title={Algebraic Geometry Codes for Secure Distributed Matrix Multiplication}, 
      author={Camilla Hollanti and Okko Makkonen and Elif Saçıkara},
      year={2023},
      eprint={2303.15429},
      archivePrefix={arXiv},
      primaryClass={cs.IT}
}

@INPROCEEDINGS{field_trace,
  author={Machado, Roberto Assis and D’Oliveira, Rafael G. L. and Rouayheb, Salim El and Heinlein, Daniel},
  booktitle={2021 XVII International Symposium "Problems of Redundancy in Information and Control Systems" (REDUNDANCY)}, 
  title={Field Trace Polynomial Codes for Secure Distributed Matrix Multiplication}, 
  year={2021},
  volume={},
  number={},
  pages={188-193},
  doi={10.1109/REDUNDANCY52534.2021.9606447}}

@ARTICLE{jie_cam_coded,
  author={Li, Jie and Hollanti, Camilla},
  journal={IEEE Transactions on Information Forensics and Security}, 
  title={Private and Secure Distributed Matrix Multiplication Schemes for Replicated or MDS-Coded Servers}, 
  year={2022},
  volume={17},
  number={},
  pages={659-669},
  doi={10.1109/TIFS.2022.3147638}}

@INPROCEEDINGS{deep_nn_coded,
  author={Dutta, Sanghamitra and Bai, Ziqian and Jeong, Haewon and Low, Tze Meng and Grover, Pulkit},
  booktitle={2018 IEEE International Symposium on Information Theory (ISIT)}, 
  title={A Unified Coded Deep Neural Network Training Strategy based on Generalized PolyDot codes}, 
  year={2018},
  volume={},
  number={},
  pages={1585-1589},
  doi={10.1109/ISIT.2018.8437852}}

@INPROCEEDINGS{lagrange,
  author={Yu, Qian and Li, Songze and Raviv, Netanel and Mousavi Kalan, Seyed Mohammadrez and Solanolkotabi, Mahdi and Avestimehr, A.\ Salman},
  title={Lagrange Coded Computing: Optimal Design for Resiliency, Security and Privacy},
  booktitle={Proceedings of the 22nd International Conference on Artificial Intelligence and Statistics (AISTATS)},
  title={},
  year={2019},
  volume={PMLR: Volume 89},
  number={},
  pages={}}

@ARTICLE{high_dim_coded_journal,
  author={Lee, Kangwook and Lam, Maximilian and Pedarsani, Ramtin and Papailiopoulos, Dimitris and Ramchandran, Kannan},
  journal={IEEE Transactions on Information Theory}, 
  title={Speeding Up Distributed Machine Learning Using Codes}, 
  year={2018},
  volume={64},
  number={3},
  pages={1514-1529},
  doi={10.1109/TIT.2017.2736066}}

@INPROCEEDINGS{high_dim_coded,
  author={Lee, Kangwook and Suh, Changho and Ramchandran, Kannan},
  booktitle={2017 IEEE International Symposium on Information Theory (ISIT)}, 
  title={High-dimensional coded matrix multiplication}, 
  year={2017},
  volume={},
  number={},
  pages={2418-2422},
  doi={10.1109/ISIT.2017.8006963}}

@ARTICLE{systematic_private,
  author={Zhu, Jinbao and Li, Songze},
  journal={IEEE Journal on Selected Areas in Information Theory}, 
  title={A Systematic Approach Towards Efficient Private Matrix Multiplication}, 
  year={2022},
  volume={3},
  number={2},
  pages={257-274},
  doi={10.1109/JSAIT.2022.3181144}}

@ARTICLE{koreans_coded,
  author={Kim, Minchul and Yang, Heecheol and Lee, Jungwoo},
  journal={IEEE Communications Letters}, 
  title={Fully Private Coded Matrix Multiplication From Colluding Workers}, 
  year={2021},
  volume={25},
  number={3},
  pages={730-733},
  doi={10.1109/LCOMM.2020.3037744}}

@ARTICLE{mds_coded_pmm,
  author={Zhu, Jinbao and Li, Songze and Li, Jie},
  journal={IEEE Transactions on Information Forensics and Security}, 
  title={Information-Theoretically Private Matrix Multiplication From MDS-Coded Storage}, 
  year={2023},
  volume={18},
  number={},
  pages={1680-1695},
  doi={10.1109/TIFS.2023.3249565}}

@INPROCEEDINGS{limited_sharings,
  author={Nodehi, Hanzaleh Akbari and Maddah-Ali, Mohammad Ali},
  booktitle={2018 IEEE International Symposium on Information Theory (ISIT)}, 
  title={Limited-Sharing Multi-Party Computation for Massive Matrix Operations}, 
  year={2018},
  volume={},
  number={},
  pages={1231-1235},
  doi={10.1109/ISIT.2018.8437651}}

@INPROCEEDINGS{tradeoff,
  author={Li, Songze and Maddah-Ali, Mohammad Ali and Avestimehr, A. Salman},
  booktitle={2016 IEEE International Symposium on Information Theory (ISIT)}, 
  title={Fundamental tradeoff between computation and communication in distributed computing}, 
  year={2016},
  volume={},
  number={},
  pages={1814-1818},
  doi={10.1109/ISIT.2016.7541612}}

@INPROCEEDINGS{k_means,
  author={Sheth, Utsav and Dutta, Sanghamitra and Chaudhari, Malhar and Jeong, Haewon and Yang, Yaoqing and Kohonen, Jukka and Roos, Teemu and Grover, Pulkit},
  booktitle={2018 IEEE International Conference on Big Data (Big Data)}, 
  title={An Application of Storage-Optimal MatDot Codes for Coded Matrix Multiplication: Fast k-Nearest Neighbors Estimation}, 
  year={2018},
  volume={},
  number={},
  pages={1113-1120},
  doi={10.1109/BigData.2018.8622429}}

@ARTICLE{secure_coded_mp,
  author={Akbari-Nodehi, Hanzaleh and Maddah-Ali, Mohammad Ali},
  journal={IEEE Transactions on Information Theory}, 
  title={Secure Coded Multi-Party Computation for Massive Matrix Operations}, 
  year={2021},
  volume={67},
  number={4},
  pages={2379-2398},
  doi={10.1109/TIT.2021.3050853}}

@ARTICLE{straggler_optimal,
  author={Dutta, Sanghamitra and Fahim, Mohammad and Haddadpour, Farzin and Jeong, Haewon and Cadambe, Viveck and Grover, Pulkit},
  journal={IEEE Transactions on Information Theory}, 
  title={On the Optimal Recovery Threshold of Coded Matrix Multiplication}, 
  year={2020},
  volume={66},
  number={1},
  pages={278-301},
  doi={10.1109/TIT.2019.2929328}}

@ARTICLE{straggler_fund,
  author={Yu, Qian and Maddah-Ali, Mohammad Ali and Avestimehr, A. Salman},
  journal={IEEE Transactions on Information Theory}, 
  title={Straggler Mitigation in Distributed Matrix Multiplication: Fundamental Limits and Optimal Coding}, 
  year={2020},
  volume={66},
  number={3},
  pages={1920-1933},
  doi={10.1109/TIT.2019.2963864}}

@inproceedings{polycodes,
author = {Yu, Qian and Maddah-Ali, Mohammad Ali and Avestimehr, A. Salman},
title = {Polynomial Codes: An Optimal Design for High-Dimensional Coded Matrix Multiplication},
year = {2017},
isbn = {9781510860964},
publisher = {Curran Associates Inc.},
address = {Red Hook, NY, USA},
pages = {4406–4416},
numpages = {11},
location = {Long Beach, California, USA},
series = {NIPS'17}
}

@INPROCEEDINGS{rawad_latency,
  author={Bitar, Rawad and Parag, Parimal and El Rouayheb, Salim},
  booktitle={2017 IEEE International Symposium on Information Theory (ISIT)}, 
  title={Minimizing latency for secure distributed computing}, 
  year={2017},
  volume={},
  number={},
  pages={2900-2904},
  doi={10.1109/ISIT.2017.8007060}}

@INPROCEEDINGS{rafael_note,
  author={D’Oliveira, Rafael G. L. and Rouayheb, Salim El and Heinlein, Daniel and Karpuk, David},
  booktitle={2020 IEEE Conference on Communications and Network Security (CNS)}, 
  title={Notes on Communication and Computation in Secure Distributed Matrix Multiplication}, 
  year={2020},
  volume={},
  number={},
  pages={1-6},
  doi={10.1109/CNS48642.2020.9162296}}

@INPROCEEDINGS{upload_vs_download,
  author={Chang, Wei-Ting and Tandon, Ravi},
  booktitle={2019 IEEE Information Theory Workshop (ITW)}, 
  title={On the Upload versus Download Cost for Secure and Private Matrix Multiplication}, 
  year={2019},
  volume={},
  number={},
  pages={1-5},
  doi={10.1109/ITW44776.2019.8989342}}

@INPROCEEDINGS{bivariate,
  author={Hasircioglu, Burak and G\'omez-Vilardeb\'o, Jes\'us and G\"und\"uz, Deniz},
  booktitle={2021 IEEE International Symposium on Information Theory (ISIT)}, 
  title={Speeding Up Private Distributed Matrix Multiplication via Bivariate Polynomial Codes}, 
  year={2021},
  volume={},
  number={},
  pages={1853-1858},
  doi={10.1109/ISIT45174.2021.9517739}}

@ARTICLE{rawad_adaptive,
  author={Bitar, Rawad and Xhemrishi, Marvin and Wachter-Zeh, Antonia},
  journal={IEEE Transactions on Information Theory}, 
  title={Adaptive Private Distributed Matrix Multiplication}, 
  year={2022},
  volume={68},
  number={4},
  pages={2653-2673},
  doi={10.1109/TIT.2022.3143199}}

@ARTICLE{entangled2,
  author={Yu, Qian and Maddah-Ali, Mohammad Ali and Avestimehr, A. Salman},
  journal={IEEE Transactions on Information Theory}, 
  title={Straggler Mitigation in Distributed Matrix Multiplication: Fundamental Limits and Optimal Coding}, 
  year={2020},
  volume={66},
  number={3},
  pages={1920-1933},
  doi={10.1109/TIT.2019.2963864}}

@INPROCEEDINGS{entangled1,
  author={Yu, Qian and Avestimehr, A. Salman},
  booktitle={2020 IEEE International Symposium on Information Theory (ISIT)}, 
  title={Entangled Polynomial Codes for Secure, Private, and Batch Distributed Matrix Multiplication: Breaking the ``Cubic'' Barrier}, 
  year={2020},
  volume={},
  number={},
  pages={245-250},
  doi={10.1109/ISIT44484.2020.9174167}}

@ARTICLE{jafar_batch,
  author={Jia, Zhuqing and Jafar, Syed Ali},
  journal={IEEE Transactions on Information Theory}, 
  title={Cross Subspace Alignment Codes for Coded Distributed Batch Computation}, 
  year={2021},
  volume={67},
  number={5},
  pages={2821-2846},
  doi={10.1109/TIT.2021.3064827}}

@ARTICLE{flex,
  author={Aliasgari, Malihe and Simeone, Osvaldo and Kliewer, J\"org},
  journal={IEEE Transactions on Information Forensics and Security}, 
  title={Private and Secure Distributed Matrix Multiplication With Flexible Communication Load}, 
  year={2020},
  volume={15},
  number={},
  pages={2722-2734},
  doi={10.1109/TIFS.2020.2972166}}

@ARTICLE{inner_product,
  author={Mital, Nitish and Ling, Cong and G\"und\"uz, Deniz},
  journal={IEEE Transactions on Information Theory}, 
  title={Secure Distributed Matrix Computation With Discrete Fourier Transform}, 
  year={2022},
  volume={68},
  number={7},
  pages={4666-4680},
  doi={10.1109/TIT.2022.3158868}}

@ARTICLE{kakar,
  author={Kakar, Jaber and Ebadifar, Seyedhamed and Sezgin, Aydin},
  journal={IEEE Access}, 
  title={On the Capacity and Straggler-Robustness of Distributed Secure Matrix Multiplication}, 
  year={2019},
  volume={7},
  number={},
  pages={45783-45799},
  doi={10.1109/ACCESS.2019.2908024}}

@INPROCEEDINGS{ravi,
  author={Chang, Wei-Ting and Tandon, Ravi},
  booktitle={2018 IEEE Global Communications Conference (GLOBECOM)}, 
  title={On the Capacity of Secure Distributed Matrix Multiplication}, 
  year={2018},
  volume={},
  number={},
  pages={1-6},
  doi={10.1109/GLOCOM.2018.8647313}}

@INPROCEEDINGS{root_of_unity,
  author={Machado, Roberto Assis and Manganiello, Felice},
  booktitle={2022 IEEE Information Theory Workshop (ITW)}, 
  title={Root of Unity for Secure Distributed Matrix Multiplication: Grid Partition Case}, 
  year={2022},
  volume={},
  number={},
  pages={155-159},
  doi={10.1109/ITW54588.2022.9965858}}

@Article{oliver,
AUTHOR = {Byrne, Eimear and Gnilke, Oliver W. and Kliewer, J\"org},
TITLE = {Straggler- and Adversary-Tolerant Secure Distributed Matrix Multiplication Using Polynomial Codes},
JOURNAL = {Entropy},
VOLUME = {25},
YEAR = {2023},
NUMBER = {2},
ARTICLE-NUMBER = {266},
URL = {https://www.mdpi.com/1099-4300/25/2/266},
PubMedID = {36832632},
ISSN = {1099-4300},
DOI = {10.3390/e25020266}
}

@ARTICLE{d2021degree,
  author={D’Oliveira, Rafael G. L. and El Rouayheb, Salim and Heinlein, Daniel and Karpuk, David},
  journal={IEEE Journal on Selected Areas in Information Theory}, 
  title={Degree Tables for Secure Distributed Matrix Multiplication}, 
  year={2021},
  volume={2},
  number={3},
  pages={907-918},
  doi={10.1109/JSAIT.2021.3102882}}

@ARTICLE{d2020gasp,
  author={D’Oliveira, Rafael G. L. and El Rouayheb, Salim and Karpuk, David},
  journal={IEEE Transactions on Information Theory}, 
  title={GASP Codes for Secure Distributed Matrix Multiplication}, 
  year={2020},
  volume={66},
  number={7},
  pages={4038-4050},
  doi={10.1109/TIT.2020.2975021}}

@online{CoCalc,
  author = {{Sagemath, Inc.}},
  title = {{CoCalc – Collaborative Calculation and Data Science}},
  year = 2023,
  note = {\url{https://cocalc.com}},
}
}

\appendix

To prove both Theorems \ref{mpc_rate} and \ref{ggasp_r_rate}, we proceed by first writing $\supp(h) = \supp(fg)$ as the union
\begin{equation}\label{basic_union}
\supp(h) = \supp(fg) = \supp(f_Ig_I) \cup \supp(f_Ig_R) \cup \supp(f_Rg_I) \cup \supp(f_Rg_R)
\end{equation}
and deriving from this an expression for $\supp(h)$ as a disjoint union of intervals of the form $[x:y]$.  From these expressions it is simple to compute the corresponding value of $N$.  Lemmas \ref{starting_point} and \ref{dumb_lemma} below will serve as the starting point for our proofs of Theorems \ref{mpc_rate} and \ref{ggasp_r_rate}.  We have omitted proofs of the following lemmas as they are a result of elementary calculations.

\begin{lem}\label{starting_point}
We can express each of the components in the union expression \eqref{basic_union} for $\supp(fg)$ as
\begin{align*}
    \supp(f_Ig_I) &= [0:KML+M-2] \\
    \supp(f_Ig_R) &= \bigcup_{t = 0}^{T-1} [KML + \beta_t: KML + KM - 1 + \beta_t] \\
    \supp(f_Rg_I) &= \bigcup_{t = 0}^{T-1}\bigcup_{\ell = 0}^{L-1}[KML + \ell KM + \alpha_t: KML + \ell KM + M -1 + \alpha_t] \\
    \supp(f_Rg_R) &= \{2KML + \alpha_t + \beta_s\ |\ 0\leq t,s\leq T-1\}
\end{align*}
\end{lem}

\begin{lem}\label{dumb_lemma}
    Let $a$ be a non-negative integer and let $S$ be a positive integer such that $S>aM$.  Then
    \[
    |\{i\in [aM:S]\ |\ i\equiv M-1\Mod{M}\}| = \left\lfloor\frac{S-aM+1}{M}\right\rfloor
    \]
\end{lem}

\subsection{Proof of Theorem \ref{mpc_rate}} \label{mpc_rate_proof}


We will only give the proof in the case that $r\nmid T$, with the proof in the case of $r|T$ being nearly identical.

The proof proceeds by writing $\supp(h)$ as a disjoint union of intervals and then using Lemma \ref{dumb_lemma} to calculate the number of integers $i\equiv M-1\Mod{M}$ contained in each interval.  Under the assumption that $\alpha_t = \beta_t = tD$ and $D\leq M$, we can rewrite the expressions in Lemma \ref{starting_point} as
\begin{align*}
\supp(f_Ig_I) \cup \supp(f_Ig_R) &= [0 : KML + KM - 1 + (T-1)D] \\
\supp(f_Rg_I) &= \bigcup_{\ell = 0}^{L - 1}[KML + \ell KM : KML + \ell KM + M - 1 + (T-1)D] \\
\supp(f_Rg_R) &= \{2KML + tD\ |\ 0\leq t\leq 2T-2\}
\end{align*}
Define an integer $\ell_0$ as in the statement of the theorem.  Then if $\ell_0 < L - 1$, we have that $\ell_0$ is the unique integer satisfying
\[
KML + \ell_0 KM \leq KML + KM - 1 + (T-1)D < KML + (\ell_0 + 1) KM
\]
Then $\ell_0$ captures how many of the intervals appearing in $\supp(f_Rg_I)$ are overlapped by the interval $\supp(f_Ig_I)\cup\supp(f_Ig_R)$.

We can decompose any set of the form $\bigcup_{i=i_0}^{I-1} [x_i:y_i]$ where $x_i\leq y_i$, $x_i\leq x_{i+1}$, and $y_i\leq y_{i+1}$ as a disjoint union by writing
\[
\bigcup_{i = i_0}^{I - 1}[x_i:y_i] = \bigsqcup_{i = i_0}^{I-2}[x_i:\min(y_i,x_{i+1}-1)] \sqcup [x_{I-1}:y_{I-1}]
\]
Applying this fact to the above and defining $k_0 = \min(M + (T-1)D, KM)$, we can write $\supp(h)$ in its entirely as the following union whose subsets are all disjoint, excluding the final term:
\begin{align*}
\supp(h) &= [0:KML + \ell_0KM + k_0] \sqcup \bigsqcup_{\ell = \ell_0 + 1}^{L-2}[KML + \ell KM : KML + \ell KM + k_0] \\
&\sqcup [2KML - KM : 2KML - KM + M - 1 + (T-1)D] \\
&\cup \{2KML + tD\ |\ 0\leq t\leq 2T-2\}
\end{align*}

To express $\supp(h)$ as a completely disjoint union, it remains to understand which elements of $\supp(f_Rg_R) = \{2KML + tD\ |\ 0\leq t \leq 2T-2\}$ can appear in the other subsets.  As every other subset of $\supp(h)$ only contains integers $< 2KML$, the entirety of the intersection of $\supp(f_Rg_R)$ with all other subsets is simply the set
\[
\{2KML + tD\ |\ 0\leq t\leq 2T-2\} \cap [2KML - KM : 2KML - KM + M - 1 + (T-1)D]
\]
An integer of the form $2KML + tD$ with $t\geq 0$ lies outside this latter interval exactly when
\[
2KML + tD \geq  2KML - KM + M - 1 + (T-1)D + 1,
\]
or equivalently $t\geq t_0$.  Finally, we use the above to write $\supp(h)$ as the following completely disjoint union:
\begin{align*}
\supp(h) &= [0:KML + \ell_0KM + k_0 - 1] \sqcup \bigsqcup_{\ell = \ell_0 + 1}^{L-2}[KML + \ell KM : KML + \ell KM + k_0 - 1] \\
&\sqcup [2KML - KM : 2KML - KM + M - 1 + (T-1)D] \\
&\sqcup \{2KML + tD\ |\ t_0\leq t\leq 2T-2\}
\end{align*}
To compute $|\supp(\widehat{h})|$ we can apply Lemma \ref{dumb_lemma} to every term in the above, except for the final term which must be dealt with separately.

To compute the contribution of this last set to $\supp(\widehat{h})$, we first observe that it is equal to
\[
\{z\in [Z_1:Z_2]\ |\ z\equiv 2KML\Mod{D}\ \ \text{and}\ \ z\equiv M-1\Mod{M}\}
\]
where $Z_1 = 2KML + t_0D$ and $Z_2 = 2KML + (2T-2)D$.  By the Chinese Remainder Theorem, since $\gcd(D,M) = 1$, the modularity conditions are equivalent to a single condition $z\equiv y \Mod{DM}$ for some explicitly computable $y$.  Then the number of integers $z$ satisfying this single congruence in the set $[Z_1:Z_2]$ can be easily computed to be
\begin{equation}\label{delta}
\left\lfloor \frac{Z_2-Z_1+1}{DM}\right\rfloor + \delta = \left\lfloor \frac{(2T-2-t_0)D + 1}{DM}\right\rfloor + \delta \quad\text{where}\quad \delta\in \{0,1\}
\end{equation}
This completes the proof of the theorem.

\subsection{Proof of Theorem \ref{ggasp_r_rate}} \label{ggasp_proof}
The proof of the theorem proceeds similarly to that of Theorem \ref{mpc_rate}, in that we express $\supp(h)$ as a disjoint union of intervals of the form $[x:y]$.  We will prove the theorem under the assumption that $r_0 > 0$, and leave the case of $r_0 = 0$ as an exercise for the reader.  It is a straightforward exercise using Lemma \ref{basic_union} to show that for GGASP codes we have
    \begin{align*}
    \supp(f_Ig_I)\cup \supp(f_Ig_R) &= [0:KML + KM + T - 2] \\
    \supp(f_Rg_I)\cup \supp(f_Rg_R) &= \bigcup_{\ell = 0}^{L + U}[KML + \ell KM : KML + \ell KM + S_\ell - 1]
    \end{align*}
    where the quantities $S_\ell$ are as in the statement of the theorem.  As in the proof of Theorem \ref{mpc_rate} we can rewrite the above expression for $\supp(f_Rg_I)\cup \supp(f_Rg_R)$ as the following \emph{disjoint} union:
    \begin{align*}
     \supp(f_Rg_I)\cup \supp(f_Rg_R) &= \bigsqcup_{\ell = 0}^{L+U-1}[KML + \ell KM: KML + \ell KM + \min(S_\ell, KM - 1)] \\
    &\sqcup [2KML + UKM : 2KML + UKM + S_{L+U} - 1]
    \end{align*}
    One now argues as in the proof of Theorem \ref{mpc_rate} to show that
    \begin{align*}
    \supp(h) = \supp(fg) &= [0:KML + \max(\ell_0KM + \min(S_{\ell_0}-1,KM-1), KM + T - 2)]  \\
    &\sqcup\bigsqcup_{\ell = \ell_0+1}^{L+U-1}[KML + \ell KM: KML + \ell KM + \min(S_\ell-1, KM - 1)] \\
    &\sqcup [2KML + UKM : 2KML + UKM + S_{L+U} - 1]
    \end{align*}
    The theorem now follows immediately by summing the sizes of the individual components of the above disjoint union.

\end{document}